\documentclass[a4paper, 11pt]{article}
\usepackage{srcltx,graphicx}
\usepackage{amsmath, amssymb, amsthm}
\usepackage{color}
\usepackage{lscape}
\usepackage{multirow}
\usepackage{psfrag}
\usepackage{hyperref}
\usepackage[hang]{subfigure}

\newtheorem{theorem}{Theorem}

\newtheorem{definition}{Definition}

\newcommand\bbR{\mathbb{R}}
\newcommand\bbN{\mathbb{N}}
\newcommand\bxi{\boldsymbol{\xi}}
\newcommand\bv{\boldsymbol{v}}
\newcommand\bx{\boldsymbol{x}}

\newcommand\bu{\boldsymbol{u}}

\newcommand\bdeta{\boldsymbol{\eta}}
\newcommand\bA{\boldsymbol{\mathrm{A}}}

\newcommand\bPhi{\boldsymbol{\phi}}
\newcommand\dd{\,\mathrm{d}}

\newcommand\mh{\mathcal{H}}
\newcommand\mH{\mathcal{H}^{[u, \theta]}}
\newcommand\mbH{\mathcal{H}^{[\bu, \theta]}}
\newcommand\mHT{\mathcal{H}^{[\bu, \Theta]}}
\newcommand\mN{\mathcal{N}}

\newcommand\bw{\boldsymbol{w}}
\newcommand\bdf{\boldsymbol{f}}
\newcommand\bmH{\boldsymbol{\mathcal{H}}^{[u,\theta]}}
\newcommand\bD{\boldsymbol{\mathrm{D}}}
\newcommand\bM{\boldsymbol{\mathrm{M}}}
\newcommand\bQ{\boldsymbol{\mathrm{Q}}}
\newcommand\bS{\boldsymbol{S}}
\newcommand\weight{\omega}
\newcommand\bbH{\mathbb H}
\newcommand\mP{\mathcal{P}}
\newcommand\bPb{\boldsymbol{\mathrm{P}}_b}
\newcommand\bPp{\boldsymbol{\mathrm{P}}_p}
\newcommand\tbPb{\boldsymbol{\mathrm{\tilde{P}}}_b}
\newcommand\tbPp{\boldsymbol{\mathrm{\tilde{P}}}_p}
\newcommand\mL{\mathcal{L}}
\newcommand\identity{\boldsymbol{\mathrm{I}}}

\newcommand\paperauthor[1]{\textsc{#1}}
\newcommand\paperauthors[1]{\textsc{#1} et al.}
\newcommand\Grad{{Grad's }}
\newcommand\rmspan{\mathrm{span}}

\newcommand\pd[2]{\dfrac{\partial {#1}}{\partial {#2}}}
\newcommand\od[2]{\dfrac{\dd {#1}}{\dd {#2}}}
\newcommand\odd[2]{\dfrac{\mathrm{D} {#1}}{\mathrm{D} {#2}}}

\numberwithin{equation}{section}

\graphicspath{{images/}}

{\theoremstyle{remark} \newtheorem{remark}{Remark}}

\title{Model Reduction of Kinetic Equations by Operator Projection}

\author{Yuwei Fan\thanks{School of Mathematical Sciences, Peking University,
    Beijing, China, email: {\tt ywfan@pku.edu.cn}.},~~
    Julian Koellermeier\thanks{Center for Computational Engineering Science,
    RWTH Aachen University, Aachen, Germany,
    email: {\tt koellermeier@mathcces.rwth-aachen.de}},~~
    Jun Li\thanks{School of Mathematical Sciences, Peking University,
    Beijing, China, email: {\tt lijun609@pku.edu.cn}.},~~
    Ruo Li\thanks{CAPT, LMAM \& School of Mathematical Sciences, Peking
    University, Beijing, China, email: {\tt rli@math.pku.edu.cn}.},~~
    Manuel Torrilhon\thanks{Center for Computational Engineering Science,
    RWTH Aachen University, Aachen, Germany,
    email: {\tt mt@mathcces.rwth-aachen.de}}~~
}

\begin{document}
\maketitle

\begin{abstract}
  By a further study of the mechanism of the hyperbolic regularization
  of the moment system for the Boltzmann equation proposed in \cite{Fan},
  we point out that the key point is treating the time and space
  derivative in the same way. Based on this understanding, a uniform
  framework to derive globally hyperbolic moment systems from kinetic
  equations using an operator projection method is proposed. The
  framework is so concise and clear that it can be treated as an
  algorithm with four inputs to derive hyperbolic moment systems by
  routine calculations. Almost all existing globally hyperbolic moment
  systems can be included in the framework, as well as some new moment
  system including globally hyperbolic regularized versions of Grad's
  ordered moment systems and a multi-dimensional extension of the
  quadrature-based moment system.

\vspace*{4mm}
\end{abstract}

\section{Introduction}
Kinetic equations, such as the Boltzmann equation and the radiative
transfer equation, are widely used in many different fields of
applications, including rarefied gases, microflow, semi-conductor
device simulation, radiative transfer, and so on. During the past
decades, various solution methods have been developed to investigate
kinetic equations. Among these methods, the moment method is quite
attractive due to its numerous advantages \cite{Muller,
  Struchtrup2002, TorrilhonEditorial}, and it is regarded as a
successful tool to extend classical fluid dynamics, and achieve highly
accurate approximations with great efficiency.

The moment method for gas kinetic theory was first proposed by
\paperauthor{Grad} \cite{Grad} in 1949, and the most notable \Grad 13
moment system is also proposed therein. In the same paper, the moment
system has been carefully studied, including the characteristics. 
Although the loss of hyperbolicity of the moment system was not
pointed out by Grad himself, it is not hard to observe it from his
paper.  Later, in \cite{Muller} it was pointed out that the 1D
reduction of \Grad 13 moment
system is only hyperbolic around the Maxwellian, and in
\cite{Grad13toR13} it was further revealed that for the 3D case, the
moment system is not hyperbolic even in any neighbourhood of the
Maxwellian. Since the hyperbolicity cannot be guaranteed for \Grad
moment method, the moment system as a quasi-linear partial
differential system with Cauchy data is no longer well-posed even
locally. Hence, the application of the moment method was seriously
limited for a long time. However, some research in recent years
brought new hope for this problem. \paperauthor{Levermore}
proposed the maximum entropy method \cite{Levermore} in 1995, and his
method yields globally hyperbolic equations but can unfortunately not
be derived in analytical form for most cases. Based on the maximum
entropy principle, \paperauthors{McDonald} \cite{McDonald} proposed an
approximative affordable robust version of Levermore's 14 moment
system, which is almost globally hyperbolic. A different hyperbolic approach that is tailored to special cases uses a multi-variate Pearson-IV-Distribution and was proposed by \paperauthor{Torrilhon} in \cite{Torrilhon2010}. Moreover, a viscous
regularization has been used to regularize \Grad moment method,
e.g. \cite{Struchtrup2003, Gu, NRxx_new}, based on the \emph{order-of-magnitude} approach also used in \cite{KaufTorrilhon}.

Concerning the global hyperbolicity of Grad-type moment systems, some
new methods are in process. The method for the 1D Boltzmann equation
introduced by \paperauthors{Cai} in \cite{Fan} is based on investigating the
properties of the coefficient matrix of the moment system. The method
therein essentially cuts off higher order terms during the derivation
such that it is globally hyperbolic. Then the method is extended to
the multi-dimensional case in two different ways \cite{Fan_new,
  ANRxx}. Shortly thereafter, \paperauthor{Koellermeier} proposed a
quadrature-based regularization method
\cite{KoellermeierMSc2013}. This method deduces the moment model by
computing the integrals using a suitable quadrature rule instead of
exact integration. This results in a globally hyperbolic moment system
very similar to the one given in \cite{Fan}. The method had since been further
extended to the multi-dimensional case in \cite{KoellermeierRGD2014},
though the resulting system is not rotational invariant. Both methods
in \cite{Fan} and \cite{KoellermeierMSc2013} have been extended to
more general cases in \cite{framework} and \cite{Koellermeier},
respectively, which has led to a better understanding of the
hyperbolicity of moment systems and the corresponding regularizations.

Based on the understanding of these new methods, in this paper we
focus on a general framework to cover all of the different
methods. We begin with the investigation of the globally hyperbolic
moment equations (HME) proposed in \cite{Fan} and point out that the
key point of the regularization is treating the time and space
derivative in the same way. Based on this understanding, by
considering different kinds of kinetic equations, a general framework
to deduce globally hyperbolic moment systems is proposed using an
operator projection method. In this framework, the cut-off procedure
in \cite{framework} is extended to a general operator
projection and the kinetic equation under consideration can have a very
generic form, including for example the Boltzmann equation, the
transformed Boltzmann equation and the radiative transfer
equation. A so-called internal projection strategy is introduced to make the
method applicable to kinetic equations without standard form. The
ansatz is chosen as a weight function multiplied by a
polynomial. Based on the framework, the resulting moment system is
always rotational invariant and is usually globally hyperbolic. We
point out that the conditions to hyperbolicity are almost always
fulfilled.

The new framework can be regarded as an algorithm to derive moment
systems from kinetic equations, once the four inputs, i.e. the form of
the kinetic equation, the weight function, the projection and the internal
projection strategy, are given. The weight function in the ansatz
space determines most of the properties of the resulting system. The
choice of a suitable polynomial basis, a projection operator and an internal
projection strategy provide us with a lot of freedom to achieve different
moment systems. This makes it possible to derive a moment system with
routine calculations and allows for easy comparison of different
models. We point out that the new framework can give us most of the
traditional moment systems, such as hyperbolic moment equations (HME)
proposed in \cite{Fan, Fan_new, ANRxx}, the quadrature-based moment
equations (QBME) \cite{KoellermeierMSc2013} and Levermore's maximum
entropy method \cite{Levermore} for the Boltzmann equation as well as the $P_N$ and
$M_N$ method in radiative transfer. Moreover, one can derive totally
new moment systems based on the framework. We provide some examples
including a hyperbolic regularization of the ordered moment hierarchy
(such as 13, 26, 45 moment systems) and extend the QBME to the
multi-dimensional case with the resulting moment system being
rotational invariant.

The remaining part of this paper is organized as follows. Some
necessary notation about projection operators is given in Section
\ref{sec:preliminary} and then we analyze the hyperbolic
regularization by \paperauthors{Cai} for Grad's moment method in
Section \ref{sec:mmbe}. In Section \ref{sec:framework}, we give the
new framework with a detailed discussion. Several examples of existing
moment systems derived using our new framework are given in Section
\ref{sec:phmsf}. Finally, we derive some new hyperbolic
regularizations with the operator projection approach in Section
\ref{sec:nhmsf}. The paper ends with a conclusion.

\section{Preliminaries}
\label{sec:preliminary}
Let $\bbR^D$ be the $D$-dimensional real space. We introduce a
function $\weight$ on $\bbR^D$, which is referred to as {\it weight
  function} hereafter, satisfying
\[
    0<\omega(\bx)<\infty,\quad
    0\leq \int_{\bbR^D}\bx^{\alpha}\weight(\bx)\dd\bx <
    \infty,\quad
    \forall \alpha\in\bbN^D,
\]
where $\bx^\alpha=\prod_{d=1}^Dx_d^{\alpha_d}$. Associated with the
weight function $\weight$, we define a weighted polynomial space
$\bbH^{\weight} = \rmspan\left\langle \{\bx^{\alpha}
  \weight(\bx)\}_{\alpha \in \bbN^D} \right\rangle $, which is an
infinite-dimensional linear space equipped with the norm
\[
    (f,g)_{\weight}:=\int_{\bbR^D}\frac{1}{\weight(\bx)}
    f(\bx) g(\bx)\dd\bx,\quad f,g\in\bbH^{\weight}.
\]

For a positive integer $n\in\bbN$, let $\bbH^{\weight}_n$ be a closed
subspace of $\bbH^{\weight}$ and $\mathrm{dim}(\bbH^{\weight}_n)=n+1$.
We call the finite-dimensional space $\bbH_n^{\weight}$ an {\it
  admissible}\footnote{See discussion on the admissible subspace for
  any moment method in \cite{Levermore}.} subspace if
\begin{itemize}
    \item $\rmspan\left\langle \weight(\bx)\{1, \bx, |\bx|^2\}\right\rangle
        \subset\bbH_n^{\weight}$,
    \item if $g(\bx)\in\bbH_n^{\weight}$, then
        $g(\boldsymbol{\mathrm{Q}}\bx+\boldsymbol{b})\in\bbH_n^{\weight}$,
        where $\boldsymbol{\mathrm{Q}}$ is a rotation matrix and
        $\boldsymbol{b}$ is a translation vector.
\end{itemize}

Let $\{\phi_0, \phi_1, \dots, \phi_k, \dots\}$ be a basis of
$\bbH^{\weight}$ and $\{\varphi_0, \varphi_1, \dots, \varphi_n\}$ be
a basis of $\bbH^{\weight}_n$, respectively. Since $\bbH^{\weight}_n$
is a subspace of $\bbH^{\weight}$, there exists a matrix
$\bPb\in\bbR^{(n+1) \times \infty}$ with full row rank such that
$\boldsymbol{\varphi} = \bPb\boldsymbol{\phi}$, where
$\boldsymbol{\phi} = (\phi_0, \phi_1, \dots, \phi_k, \dots)^T$ and
$\boldsymbol{\varphi} = (\varphi_0, \varphi_1, \dots, \varphi_n)^T$.

A linear bounded operator $\mP: \bbH^{\weight}\rightarrow
\bbH^{\weight}$ is called a {\it projection operator} on
$\bbH^{\weight}_n$ if
\begin{itemize}
    \item $\mP g\in\bbH^{\weight}_n$ for all $g\in\bbH^{\weight}$,
    \item $\mP g=g$ for all $g\in\bbH^{\weight}_n$.
\end{itemize}

For any $g\in\bbH^{\weight}$, there exists $g_i,
i=0,\dots,\infty$ such that $g=\sum_{i=0}^{\infty}g_i\phi_i$.
Since $\mP g\in\bbH^{\weight}_n$, there exists
$\hat{g}_i,i=0,\dots,n$ such that $\mP
g=\sum_{i=0}^n\hat{g}_i\varphi_i$.
Since $\mP$ is a linear operator, there exists a unique matrix
$\bPp\in\bbR^{n+1\times\infty}$ satisfying
$\boldsymbol{\hat{g}}=\bPp\boldsymbol{g}$, where
$\boldsymbol{\hat{g}}=(g_0,\dots,g_n)^T$ and
$\boldsymbol{g}=(g_0,\dots,g_k,\dots)^T$, such that
\begin{equation}
    \mP g=\left\langle \bPb\boldsymbol{\phi}, \bPp\boldsymbol{g}\right\rangle_{N} ,
\end{equation}
where $\left\langle \cdot,\cdot\right\rangle_{N}$ denotes the inner
product of finite size vectors as opposed to $\left\langle
  \cdot,\cdot\right\rangle_{\infty}$ for infinite size vectors that
will be used later.  Noticing that $\mP$ is a linear bounded
projection operator, we have
\begin{equation}\label{eq:projection}
    \begin{aligned}
    ||\bPp||< \infty,\quad  \text{where $||\cdot||$ is a matrix norm,}\\
    \bPb\bPp^T=\identity_n, \quad \identity_n \text{ is an } n\times n
    \text{ identity matrix. }
    \end{aligned}
\end{equation}
Clearly the projection operator $\mP$ is uniquely determined by
$\bPp$, thus hereafter we may directly use the matrix $\bPp$ to denote
the projection on the weighted polynomial space.

Particularly, for the classical orthogonal projection, i.e.
\[
    ||\mP g-g||_{\weight}\leq||f-g||_{\weight}, \quad \forall
    f\in\bbH^{\weight},
\]
we have
\begin{equation}\label{eq:bpp}
  \bPp=\left( (\varphi_i, \varphi_j)_{\weight} \right)_{(n+1) \times
    (n+1)}^{-1} \cdot \bPb\cdot \left(
    (\phi_i, \phi_j)_{\weight} \right)_{\infty\times\infty}.
\end{equation}
Furthermore, if $\varphi_i = \phi_i$, $i = 0, \dots, n$ and
$(\varphi_i, \phi_j)_{\weight} = 0$ for all $i = 0, \dots, n$, $j =
n+1, n+2, \dots$, then the orthogonal projection is actually a cut-off
and we have
\[
\bPb=\bPp=\boldsymbol{\mathrm{T}}:=\begin{pmatrix}
  \identity _{n+1}  & \boldsymbol{0}
\end{pmatrix},
\]
where $\identity _{n+1}$ is the $(n+1)$-th order identity matrix.

For later use, we note
\begin{definition}[Hyperbolicity]
\label{def:hyperbolicity}
A system of first order quasi-linear partial differential equations
\[
\pd{\bw}{t}+\sum_{d=1}^D\bA_d(\bw)\pd{\bw}{x_d}=0
\]
is called {\it hyperbolic} in some region $\Omega$ if and only if any
linear combination of $\bA_d(\bw)$ is diagonalizable with real
eigenvalues for all $\bw\in\Omega$.
\end{definition}

\section{Moment Method for Boltzmann Equation}
\label{sec:mmbe}

In this section, we introduce the Boltzmann equation and then briefly
review \Grad moment system of arbitrary order proposed in
\cite{NRxx} together with the globally hyperbolic regularization for
the moment system in \cite{Fan, Fan_new}. At last, we give an
alternative understanding to derive the regularized moment system.

\subsection{The Boltzmann equation}
In gas kinetic theory, the motion of particles is depicted by the
mass density distribution function $f(t, \bx, \bxi)$ governed by the
Boltzmann equation
\begin{equation}\label{eq:boltzmann}
    \pd{f}{t}+\sum_{d=1}^D\xi_d\pd{f}{x_d}=S(f),
\end{equation}
where $t$ is the time variable and $\bx\in\bbR^D$ and $\bxi\in\bbR^D$
denote the position and microscopic velocity, respectively. The right
hand side of \eqref{eq:boltzmann} $S(f)$ is used to model the
interaction among particles and is beyond our concern, thus we do not
give its concrete form and simply assume $S(f_M(t,\bx,\bxi)) =
0$. Here $f_M(t,\bx,\bxi)$ is the local Maxwellian
\begin{equation*}
    f_M(t,\bx,\bxi)=\frac{\rho(t,\bx)}{\sqrt{2\pi\theta(t,\bx)}^D}\exp\left(
    -\frac{|\bxi-\bu(t,\bx)|^2}{2\theta(t,\bx)} \right).
\end{equation*}
The macroscopic density $\rho(t,\bx)$, velocity $\bu(t,\bx)$ and
temperature $\theta(t,\bx)$ are related to the distribution function
$f(t, \bx, \bxi)$ by
\begin{equation*}
    \begin{aligned}
        \rho(t,\bx) &= \int_{\bbR^D}f(t, \bx, \bxi)\dd\bxi, \\
        \rho(t,\bx)\bu(t,\bx) &= \int_{\bbR^D}\bxi f(t, \bx, \bxi)\dd\bxi, \\
        \frac{D}{2}\rho(t,\bx)\theta(t,\bx) +\frac{1}{2}\rho(t,\bx)|\bu(t,\bx)|^2&=
        \int_{\bbR^D}\frac{1}{2}|\bxi|^2f(t, \bx, \bxi)\dd\bxi.
    \end{aligned}
\end{equation*}
Multiplying the Boltzmann equation \eqref{eq:boltzmann} by $(1, \bxi,
|\bxi|^2/2)^T$ and integrating both sides over $\bbR^D$ with respect
to $\bxi$, we get the following conservation laws
\[
    \begin{aligned}
        \pd{\rho}{t} &+ \sum_{d=1}^D\pd{\rho u_d}{x_d}=0, \\
        \rho\pd{u_i}{t} &+ \sum_{d=1}^D\left( \rho
        u_d\pd{u_i}{x_d}+\pd{p_{id}}{x_d} \right) = 0,
        \quad i = 1, \dots, D, \\
        \frac{D\rho}{2}\pd{\theta}{t}&+\sum_{d=1}^D\left(
        \frac{D}{2}\rho u_d\pd{\theta}{x_d} + \pd{q_d}{x_d} \right)
        + \sum_{d=1}^D\sum_{k=1}^Dp_{kd}\pd{u_k}{x_d}=0.
    \end{aligned}
\]
Here $p_{ij}$ and $q_i$, $i,j = 1, \dots, D$ are pressure tensor and
heat flux, respectively, defined by
\[
    p_{ij}=\int_{\bbR^D}f(t, \bx, \bxi)(\xi_i-u_i)(\xi_j-u_j)\dd\bxi, \quad
    q_i=\int_{\bbR^D}f(t, \bx, \bxi)|\bxi-\bu|^2(\xi_i-u_i)\dd\bxi.
\]

\subsection{Moment method for the Boltzmann equation}
In 1949, \paperauthor{Grad} \cite{Grad} assumed that the distribution
function is close to a local Maxwellian and expanded the
distribution function $f$ into Hermite series to obtain the Grad 13
and Grad 20 moment systems. \paperauthor{Cai} and \paperauthor{Li}
\cite{NRxx} extended it to more general cases and obtained arbitrary
order moment systems. Here we first discuss the $D=1$ case and the
multi-dimensional case will be discussed in Section \ref{sec:hme}.

\subsubsection{\Grad moment method}
Let $D=1$. Following \paperauthor{Grad}, we expand the distribution
function around the Maxwellian as follows
\begin{equation}\label{eq:expansion}
    f(t,x,\xi)=\sum_{\alpha\in\bbN}f_{\alpha}(t,x)\mathcal{H}^{[u(t,x),\theta(t,x)]}_{\alpha}(\xi),
\end{equation}
where the basis function $\mH_{\alpha}(\xi)$ is a weighted Hermite
polynomial defined as
\begin{equation}\label{eq:grad-basisfunction}
    \mH_{\alpha}(\xi) = (-1)^\alpha\dfrac{\dd^\alpha}{\dd
        \xi^{\alpha}}\weight^{[u,\theta]}(\xi), \quad \alpha\geq0,\quad
    \weight^{[u,\theta]}(\xi)=\frac{1}{\sqrt{2\pi\theta}}\exp\left(
    -\frac{|\xi-u|^2}{2\theta}\right).
\end{equation}
Here we list some basic relations of the basis function
$\mH_{\alpha}(\xi)$ as following:
\begin{itemize}
    \item orthogonality relation:
        $
        \left(\mH_{\alpha}(\xi),\mH_{\beta}(\xi)\right)_{\weight^{[u,\theta]}}
        = \dfrac{\alpha!}{\theta^\alpha}\delta_{\alpha,\beta};
        $
    \item derivative relation:
        $
            \pd{\mH_{\alpha}({\xi})}{s}=\pd{u}{s}\mH_{\alpha+1}({\xi})
            +\dfrac{1}{2}\pd{\theta}{s}\mH_{\alpha+2}({\xi}), \;~ s=t,x;
        $
    \item recurrence relation:
        $
            \xi\mH_{\alpha}({\xi})=\theta\mH_{\alpha+1}(\xi)+
            u\mH_{\alpha}(\xi) + \alpha\mH_{\alpha-1}({\xi}).
        $
\end{itemize}

Using the orthogonality relation, we get the constraints
\begin{equation}\label{eq:constrain_grad}
    f_{1} = f_{2} = 0.
\end{equation}
Then substituting the expansion \eqref{eq:expansion} into the
Boltzmann equation \eqref{eq:boltzmann}, we get
\begin{align}
\label{eq:substitute_1d1}
    \pd{f}{t}&=\sum_{\alpha\in\bbN}\left( \pd{f_{\alpha}}{t}
    + f_{\alpha-1}\pd{u}{t}
    +\frac{1}{2}f_{\alpha-2}\pd{\theta}{t}
    \right)\mH_{\alpha}(\xi),\\
\label{eq:substitute_1d2}
    \xi\pd{f}{x}&=\sum_{\alpha\in\bbN}\left( \pd{f_{\alpha}}{x}
    + f_{\alpha-1}\pd{u}{x}
    +\frac{1}{2}f_{\alpha-2}\pd{\theta}{x}
    \right)\left( u\mH_{\alpha}+\theta\mH_{\alpha+1}
    +\alpha\mH_{\alpha-1}
 \right).
\end{align}
Matching the coefficients of the basis functions in
\eqref{eq:substitute_1d1} and \eqref{eq:substitute_1d2}, we obtain
\Grad moment system with infinite number of equations
\begin{equation}\label{eq:grad-equations}
    \begin{aligned}
        \pd{f_{\alpha}}{t}&+u\pd{f_{\alpha}}{x}+\theta
        \pd{f_{\alpha-1}}{x}+(\alpha+1)\pd{f_{\alpha+1}}{x}+\\
        f_{\alpha-1}\pd{u}{t}&+\left(uf_{\alpha-1}+\theta
        f_{\alpha-2}+(\alpha+1)f_{\alpha}\right)\pd{u}{x}+\\
        \dfrac{f_{\alpha-2}}{2}\pd{\theta}{t}&+
        \dfrac{1}{2}\left( u f_{\alpha-2}+\theta
        f_{\alpha-3}+(\alpha+1)f_{\alpha-1}
        \right)\pd{\theta}{x}=S_{\alpha}, \quad \alpha \geq 3.
    \end{aligned}
\end{equation}
Here $S_{\alpha}$ is obtained by expansion of the collision part
$S(f)$. Noticing \eqref{eq:constrain_grad}, we let $\bw = (f_0, u,
\theta, f_3, f_4, \dots)$, then \eqref{eq:grad-equations} can be
written as
\begin{equation}\label{eq:grad-system}
    \bD\pd{\bw}{t}+\bM\bD\pd{\bw}{x}=\bS,
\end{equation}
where the matrices $\bD$ and $\bM$ are determined from
\eqref{eq:substitute_1d1} and \eqref{eq:substitute_1d2} and
$\bS=(S_{\alpha})_{\alpha\in\bbN}$ is a vector with entries sorted by
ascending order of $\alpha$.

Choosing an integer $M\geq 2$, discarding all the governing equations
of $f_{\alpha}, |\alpha|\ge M$ and dropping all the terms
including the space derivative of $f_{\alpha}, |\alpha|\ge M$, in the
remaining equations, we obtain \Grad $M+1$ moment system in
\cite{NRxx} for the 1D case, which can be written with modified matrices and variables as
\begin{equation}\label{eq:grad-system-projected}
    \bD_M\pd{\bw_M}{t}+(\bM\bD)_M\pd{\bw_M}{x}=\bS_M.
\end{equation}
The matrices $\bD_M$ and $(\bM\bD)_M$ as well as $\bw_M$ can be derived in a different way using the following procedure.

\subsubsection{Decomposition of the deduction}
\label{sec:viewpoint_grad}
The procedure deriving \Grad moment system can be decomposed into
the following steps:
\begin{enumerate}
    \item Weight function and weighted polynomial space:
        Choose $\weight^{[u,\theta]}(\xi)$ as the weight function, and let
        the weighted polynomial space $\bbH^{\weight^{[u,\theta]}}
        =\rmspan\left\langle
        \{\mH_{\alpha}(\xi)\}_{\alpha\in\bbN}\right\rangle $.
    \item Projection operator:
        Choose an integer $M\geq2$ and let
        $\bbH^{\weight^{[u,\theta]}}_M= \rmspan\left\langle
        \{\mH_{\alpha}(\xi)\}_{\alpha\leq M}\right\rangle $. It is
        clear that the $\mH_{\alpha}(\xi)$ form an orthogonal basis of
        $\bbH^{\weight^{[u,\theta]}}$ and
        $\bPb=\boldsymbol{\mathrm{T}}$. Here \paperauthor{Grad} used a
        direct truncation of the distribution function, which
        corresponds to orthogonal projection, so we have
        $\bPp=\boldsymbol{\mathrm{T}}$.
    \item \Grad expansion: Expand the distribution
        function in the space $\bbH^{\weight^{[u,\theta]}}$
        \[
            f(t,x,\xi)=\sum_{\alpha\in\bbR}f_{\alpha}(t,x)\mH_{\alpha}(\xi)
            = \left\langle \bmH, \bdf\right\rangle_{\infty} ,
        \]
        where $\left\langle \cdot,\cdot\right\rangle_\infty $ is the
        inner product of infinite size vectors and
        $\bmH=(\mH_{\alpha}(\xi))_{\alpha\in\bbN}$ and
        $\bdf=(f_{\alpha})_{\alpha\in\bbN}$ are vectors of elements
        sorted by ascending order of $\alpha$.
    \item Constraints:
        \begin{equation}\label{eq:grad-constrain}
            f_1=f_2=0.
        \end{equation}
        So $\bw=(f_0, u,\theta, f_3,f_4,\dots)$ contains all the
        macroscopic parameters.
    \item Projection 1: Project the distribution function into
        $\bbH^{\weight^{[u,\theta]}}_M$:
        \[
            \mP f(t,x,\xi) = \left\langle \bPb\bmH, \bPp\bdf\right\rangle_{N} .
        \]
    \item Time and space derivative: for $s=t,x$
        \begin{equation}\label{eq:grad-timederivative}
            \begin{aligned}
                \pd{\mP f}{s}&=\left\langle \bPb\pd{\bmH}{s},\bPp\bdf\right\rangle_{N} +\left\langle \bPb\bmH,
                \bPp\pd{\bdf}{s}\right\rangle_{N} \\
                &=\left\langle \bPb\boldsymbol{\mathrm{C}}\bmH,
                \bPp\bdf\right\rangle_{N} +\left\langle \bPb\bmH,
                \bPp\pd{\bdf}{s}\right\rangle_{N} \\
                &=\left\langle \bmH,
                \boldsymbol{\mathrm{C}}^T\bPb^T\bPp\bdf
                +\bPb^T\bPp\pd{\bdf}{s}\right\rangle_{\infty}
                =\left\langle \bmH, \bD\bPb^T\pd{\bPp\bw}{s}\right\rangle_{\infty} .\\
            \end{aligned}
        \end{equation}
        Here $\boldsymbol{\mathrm{C}}$ is a matrix with infinite size and can
        be deduced directly from the derivative relation of the basis
        functions. The first $M+1$ rows of the matrix $\bD$ can be
        derived from
        $\boldsymbol{\mathrm{C}}^T\bPb^T\bPp\bdf+\bPb^T\bPp\pd{\bdf}{s}$ and
        $\bD$ is the same as in \eqref{eq:grad-system}.
    \item Multiplication with velocity:
        \begin{equation}\label{eq:grad-multiplyvelocity}
            \begin{aligned}
            \xi\pd{\mP f}{x}&=\left\langle \xi\bmH, \bD\bPb^T\pd{\bPp\bw}{x}\right\rangle_{\infty}
            =\left\langle \bM^T\bmH, \bD\bPb^T\pd{\bPp\bw}{x}\right\rangle_{\infty} \\
            &=\left\langle \bmH, \bM\bD\bPb^T\pd{\bPp\bw}{x}\right\rangle_{\infty}.
            \end{aligned}
        \end{equation}
        The matrix $\bM$ can be derived directly from the recurrence
        relation of the basis functions and is the same as in
        \eqref{eq:grad-system}.
    \item Projection 2: Project \eqref{eq:grad-timederivative} and
        \eqref{eq:grad-multiplyvelocity} into the space
        $\bbH^{\weight^{[\theta]}}_M$ and match the coefficients of the
        basis functions to obtain the moment system:
        \begin{equation}\label{eq:grad1d}
            \bPp\bD\bPb^T\pd{\bPp\bw}{t}+\bPp\bM\bD\bPb^T\pd{\bPp\bw}{x}=\bPp\bS.
        \end{equation}
        This finally yields \Grad $M+1$ moment
        system.
\end{enumerate}
Comparing \eqref{eq:grad1d} and \eqref{eq:grad-system}, we observe
that \Grad truncation and closure are corresponding to the projection
on the distribution function and the moment system. Actually, we can
also first obtain system \eqref{eq:grad-system} and then let
\[
    \bw_M=\bPp\bw,\quad \bD_M=\bPp\bD\bPb^T,\quad
    (\bM\bD)_M=\bPp\bM\bD\bPb^T,\quad \bS_M=\bPp\bS,
\]
to get \eqref{eq:grad-system-projected}, which is exactly the same as \eqref{eq:grad1d}.

Note, that we do not explicitly write down the matrices $\bD$ and $\bM$ here in order to shorten notation, but some examples for different cases are given in Section \ref{sec:phmsf}.

\subsection{Globally hyperbolic moment equations}
\label{sec:hme1d}
The hyperbolicity of system \eqref{eq:grad1d} requires $\bD_M$ to be
invertible and $\bD_M^{-1}(\bM\bD)_M$ to be real diagonalizable. It is easy
to check that $\bD_M$ is invertible, since $\bD_M$ is a lower triangular
matrix and its diagonal entries are all nonzero. However, in
\cite{Fan} \paperauthors{Cai} investigated the hyperbolicity of it and
concluded that for $M\geq 3$ \Grad moment system \eqref{eq:grad1d} is
only hyperbolic around the Maxwellian. A globally hyperbolic
regularization for \Grad moment system in 1D was proposed afterwards. In
\cite{framework}, \paperauthors{Cai} investigated the regularization
and gave an explanation from the viewpoint of the discrete velocity method
and based on the regularization, a generalized framework was proposed
to obtain a hyperbolic moment system based on any ansatz for the
kinetic equation.  In this subsection, we use a diagram of the regularization
proposed in \cite{Fan} to compare the treatments of time and space
derivatives for \Grad moment system and the regularization.

\begin{figure}[ht]
    \centering
    \includegraphics[width=\textwidth]{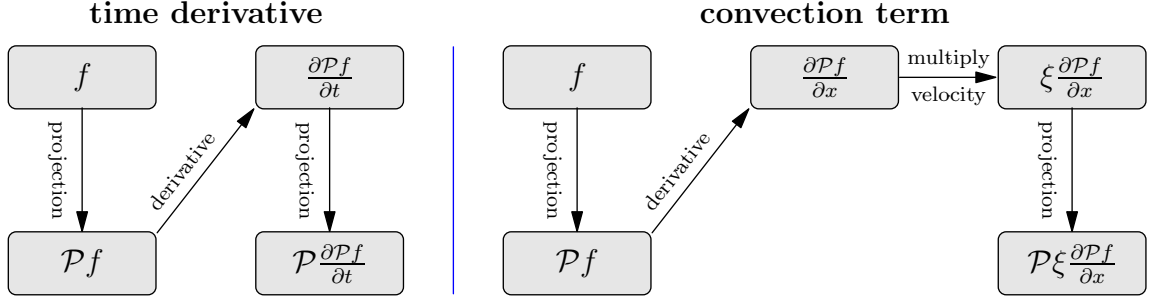}
    \caption{\label{fig:grad}
    Diagram for \Grad moment method for the 1D Boltzmann equation.
}
\end{figure}
\begin{figure}[ht]
    \centering
    \includegraphics[width=\textwidth]{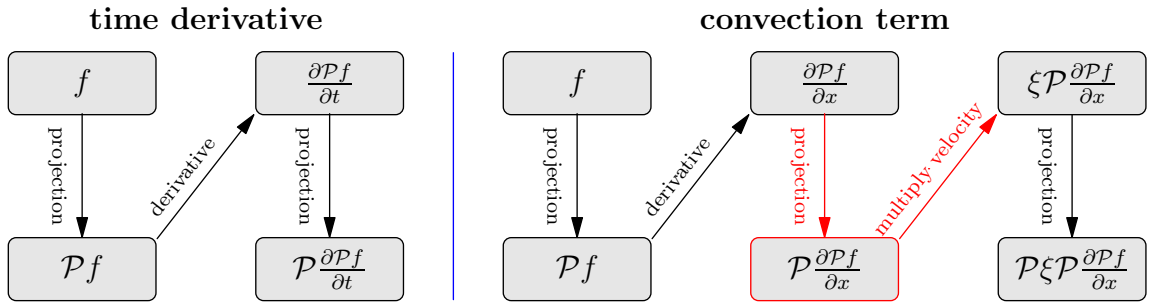}
    \caption{\label{fig:hme}
    Diagram for the regularization proposed in \cite{Fan} for the 1D
    Boltzmann equation.
}
\end{figure}

To derive \Grad moment equation, we need to calculate the time derivative
$\pd{f}{t}$ and the convection term $\xi\pd{f}{x}$. As shown
in Fig. \ref{fig:grad}, for the time derivative, the projection operator
directly acts on $\pd{\mP f}{t}$ after the time derivative. But for the
convection term, the projection operator acts on $\pd{\mP f}{x}$ after
multiplying with the velocity. That means \paperauthor{Grad} treated the time
and space derivative in different ways.
In the perspective of physics, if only the convection term is
considered in the Boltzmann equation, the system is time reversal
invariant, thus there is no essential difference for time and space.
Hence, it is natural to use the same treatment for time derivative and
space derivative. In the perspective of mathematics, the same
treatment for time and space derivatives indicates that the hyperbolicity
of the resulting moment system only depends on the operator representing the multiplication
with velocity and the hyperbolicity does not depend on the derivative operator, since
matrix similarity transformation preserves the matrix eigenstructure.
In fact, the hyperbolicity of the Boltzmann equation can be expanded since
the multiplication operator $\xi\cdot$ is real-valued, symmetric, and does
not depend on the time and space derivatives. In conclusion, it is a
natural choice to use the same treatment for the time and space
derivatives, as is shown in Fig. \ref{fig:hme}, which results in the
regularization proposed by \paperauthors{Cai} in \cite{Fan}.  Based on
the perspective in Fig. \ref{fig:hme}, the derivation of the
regularized moment system in \cite{Fan} can be written as
\begin{enumerate}
    \item[1.-6.] the same as the $1$st-$6$th step in Section
        \ref{sec:viewpoint_grad}.
\setcounter{enumi}{6}
    \item Projection 2: Project the space derivative
        \eqref{eq:grad-timederivative} into space
        $\bbH^{\weight^{[\theta]}}_M$
        \[
            \mP \pd{\mP f}{x}=\left\langle \bPb\bmH, \bPp\bD\bPb^T\pd{\bPp\bw}{x}\right\rangle_{N} .
        \]
    \item Multiplication with velocity:
        \begin{equation}\label{eq:grad-multiplyvelocity2}
            \begin{aligned}
            \xi\mP\pd{\mP f}{x}& 
            =\left\langle \bmH, \bM\bPb^T\bPp\bD\bPb^T\pd{\bPp\bw}{x}\right\rangle_{\infty}.
            \end{aligned}
        \end{equation}
    \item Projection 3: Project \eqref{eq:grad-timederivative} and
        \eqref{eq:grad-multiplyvelocity2} into the space
        $\bbH^{\weight^{[\theta]}}_M$ and match the coefficients of the
        basis functions to obtain the regularized moment system:
        \begin{equation}\label{eq:hme1d}
            \bPp\bD\bPb^T\pd{\bPp\bw}{t}+\bPp\bM\bPb^T\bPp\bD\bPb^T\pd{\bPp\bw}{x}=\bPp\bS.
        \end{equation}
        This finally yields the globally hyperbolic moment equations
        proposed in \cite{Fan}.
\end{enumerate}
Similar to \Grad moment system, we can first obtain system
\eqref{eq:grad-system} and then let
\[
    \bw_M=\bPp\bw,\quad \bD_M=\bPp\bD\bPb^T,\quad
    \bM_M=\bPp\bM\bPb^T,\quad \bS_M=\bPp\bS,
\]
to get
\begin{equation}
    \bD_M\pd{\bw_M}{t}+\bM_M\bD_M\pd{\bw_M}{x}=\bS_M.
\end{equation}
The upper system is exactly the same as \eqref{eq:hme1d}. That means
we can derive the moment system with infinite equations first without
considering the projection and then apply the projection to it to obtain
the corresponding equations. This observation will help us to
understand the difference between \Grad 13 moment system and \Grad 20
moment system, as well as the regularized versions of them.

\section{Generic Kinetic Equations}
\label{sec:framework}
In the last section, we investigated the regularized moment system proposed
in \cite{Fan}. In this section, we deduce and summarize the
characteristic of the regularization and extend it to a framework.
Based on the framework, different moment systems can be derived by
some routine calculations once the kinetic equation, the weight
function, the projection and the internal projection strategy are given and new
moment systems can be derived without essential difficulty. The
framework will be introduced step by step in this section. First, we
clarify the form of the kinetic equation.

\subsection{The form of the kinetic equation}
\label{sec:form_kineticequation}
It is natural to determine the kinetic equation before deducing the
moment system. We want to cover different kinetic equations in our framework and thus assume the following form of the kinetic equation
\begin{equation}\label{eq:kineticequation}
    \mL\left(\pd{}{t}; f, \bdeta_1, \bv(\bxi)\right) + \sum_{d=1}^D
    p_d(\bv(\bxi))\mL\left(\pd{}{x_d};f,\bdeta_1,\bv(\bxi)\right)=S(f),
\end{equation}
where $f=f(t,\bx,\bv)$, $\bdeta_1=\bdeta_1(t,\bx)$ is a vector of
macroscopic parameters \footnote{$\bdeta_1$ can be treated as a set,
but uniqueness demands that every element of $\bdeta_1$ cannot be
expressed by the others. For example, $\{\rho, \theta,p\}$ is not
allowed because $p=\rho\theta$, while $\{\rho,u,\theta\}$ is allowed.}
and $\mL\left( \pd{}{s};\cdot,\cdot,\cdot \right)$, for $s=t,x_d$ is an
operator. Furthermore $\bv(\bxi)$ is a function of $\bxi$ and
$p_d(\cdot)$ is a polynomial, which suffices to cover all major models.
Among others, the following important models are readily included in our framework:
\begin{itemize}
    \item Conventional Boltzmann equation \eqref{eq:boltzmann}:
    The standard Boltzmann equation is easily included in the framework by setting
        \[
            \bdeta_1=\emptyset,\quad
            \bv(\bxi)=\bxi,\quad p_d(\bv)=v_d,\quad
            \mL\left( \pd{}{s};f, \bdeta_1, \bv(\bxi)\right)=\pd{f}{s}, \; s=t,x_d.
        \]
    \item Scaled Boltzmann equation used in \cite{Koellermeier}:

        A transformed Boltzmann equation is obtained after shifting the
        microscopic velocity $\bxi$ by its macroscopic velocity $\bu$ and
        scaling by the standard deviation $\sqrt{\theta}$ to get a Galilean
        invariant variable transformation:
        \[
            \bxi\rightarrow\frac{\bxi-\bu}{\sqrt{\theta}}=:\bv.
        \]
        With this transformation, the Boltzmann equation \eqref{eq:boltzmann}
        is transformed to
        \begin{equation}\label{eq:scaled_boltzmann}
            \begin{aligned}
            \odd{f}{t}+\sum_{d=1}^D\sqrt{\theta}v_d\pd{f}{x_d}
            +\sum_{k=1}^D\pd{f}{v_k}{\Bigg(}
            &-\frac{1}{\sqrt{\theta}}\left(
            \odd{u_k}{t}+\sum_{d=1}^D\sqrt{\theta}v_d\pd{u_k}{x_d} \right)\\
            &-\frac{1}{2\theta}v_k\left(
            \odd{\theta}{t}+\sum_{d=1}^D\sqrt{\theta}v_d\pd{\theta}{x_d} \right)
            {\Bigg)}=S(f),
            \end{aligned}
        \end{equation}
        where the material derivative
        $\odd{}{t}:=\pd{}{t}+\sum_{d=1}^Du_d\pd{}{x_d}$ is used. In physical
        perspective, \eqref{eq:scaled_boltzmann} and \eqref{eq:boltzmann}
        depict the same physical process. In mathematical perspective, however, we
        treat the two equations as different models.

        We can include the transformed Boltzmann equation \eqref{eq:scaled_boltzmann} in our framework by setting
        \[
            \begin{aligned}
            \bdeta_1&=(u_1,\dots,u_D,\theta),\quad
            \bv(\bxi)=\frac{\bxi-\bu}{\sqrt{\theta}},\quad
            p_d(\bv)=u_d+\sqrt{\theta}v_d,\quad\\
            \mL\left( \pd{}{s};f, \bdeta_1, \bv(\bxi) \right)&=\pd{f}{s}-\sum_{k=1}^D\pd{f}{v_k}
            \left( \frac{1}{\sqrt{\theta}}\pd{u_k}{s}
            +\frac{1}{2\theta}v_k\pd{\theta}{s}
            \right),\; s=t, x_d.
            \end{aligned}
        \]
    \item Radiative transfer equation:

        The radiative transfer equation reads
        \begin{equation}\label{eq:radiativetransfer}
            \frac{1}{c}\pd{f}{t}+\bv(\bxi)\cdot\nabla f=S(f;T),
        \end{equation}
        where $c$ is the speed of light and $S(f;T)$ models interactions
        between photons and the background medium with material temperature
        $T$ and $\bv(\bxi)=\bxi/|\bxi|$.
        The radiative transfer equation \eqref{eq:radiativetransfer} is included in the framework by setting
        \[
            \bdeta_1=\emptyset,\quad \bv(\bxi)=\bxi/|\bxi|,\quad
            p_d(\bv)=c v_d,\quad
            \mL\left( \pd{}{s}; f, \bdeta_1, \bv(\bxi)\right)=\frac{1}{c}\pd{f}{s},\; s=t,x_d.
        \]
\end{itemize}
In this paper, we are not confined to the upper three cases, but
consider any kinetic equation of the form as
\eqref{eq:kineticequation}.

\subsection{The framework of model reduction}
\label{sec:framework_details}
Based on the form of the kinetic equation \eqref{eq:kineticequation},
we give a framework to derive a moment system from the kinetic equation.

\begin{enumerate}
    \item Weight function and weighted polynomial space:
        Denote the weight function by $\weight^{[\bdeta_2]}(\bv)$,
        where $\bdeta_2=\bdeta_2(t,\bx)$ is a set of some macroscopic
        parameters. Then the weighted polynomial space is
        $\bbH^{\weight^{[\bdeta_2]}}
        =\rmspan\left\langle \{\weight^{[\bdeta_2]}(\bv)\bv^\alpha\}_{\alpha\in\bbN^D}\right\rangle $,
        and let $\bPhi=(\phi_0,\phi_1,\dots,)^T$ be a basis.
    \item Projection operator: Choose an admissible subspace
        $\bbH^{\weight^{[\bdeta_2]}}_{sub}$ of
        $\bbH^{\weight^{[\bdeta_2]}}$ and determine the projection
        $\mP$, which means determining the two matrices $\bPb$ and
        $\bPp$.
    \item Ansatz: Expand the distribution function $f(t,\bx,\bv)$
        in the space $\bbH^{\weight^{[\bdeta_2]}}$
        \begin{equation}\label{eq:expansion_frame}
            f(t,\bx,\bv)=\sum_{\alpha\in\bbN^D}f_{\alpha}(t,\bx)\phi_{\alpha}(\bv)
            =\left\langle \bPhi, \bdf\right\rangle_{\infty} .
        \end{equation}
    \item Constraints: Denote $\bdeta=\bdeta_1\cup\bdeta_2$ and let
        $n$ be the cardinality of $\bdeta$. Then there must be $n$
        independent relations between $\bdeta$ and $\bdf$
        \begin{equation}\label{eq:constrain_frame}
            r_j(\bdeta,\bdf)=0,\quad j=1,\dots,n.
        \end{equation}
        Using \eqref{eq:constrain_frame} to eliminate $n$ parameters
        in $\bdeta, \bdf$, we denote the remaining by $\bw$.
    \item Projection 1: Project the distribution function into the
        space $\bbH^{\weight^{[\bdeta_2]}}_{sub}$
        \begin{equation}
            \mP f(t,\bx,\bv)=\left\langle \bPb\bPhi,\bPp\bdf\right\rangle_{N} .
        \end{equation}
    \item Time and space derivative: For $s=t, x_d$, calculate
        $\mL\left( \pd{}{s},\dots \right)$
        with an internal projection strategy $PS_1$
        \begin{equation}\label{eq:derivative_frame}
            \mL\left( \pd{}{s},\dots
            \right)\rightarrow\mL^{PS_1}\left( \pd{}{s},\dots
            \right)=\left\langle \bPhi,
            \bD_{PS_1}\bPb^T\pd{\bPp\bw}{s}\right\rangle_{\infty} ,
        \end{equation}
        where $\bD_{PS_1}$ depends on $\mL\left( \pd{}{s},\dots
        \right)$ and the internal projection strategy. In deriving $\bD_{PS_1}$,
        the projection may be used, and Section \ref{sec:QBME} gives an
        example.
    \item Projection 2: Project the resulting time and space derivative into the space
        $\bbH^{\weight^{[\bdeta_2]}}_{sub}$
        \begin{equation}\label{eq:projection2_frame}
            \mP \mL^{PS_1}\left( \pd{}{s},\dots \right)= \left\langle \bPb\bPhi,
            \bPp\bD_{PS_1}\bPb^T\pd{\bPp\bw}{s}\right\rangle_{N} .
        \end{equation}
    \item Multiplication with velocity: For $d=1,\dots,D$, calculate
        $p_d(\bv)\mP\mL^{PS_1}\left( \pd{}{s},\dots \right)$ with an
        internal projection strategy $PS_2$
        \begin{equation}\label{eq:multiplyvelocity_frame}
            p_d(\bv)\mP\mL^{PS_1}\left( \pd{}{s},\dots \right)\rightarrow
            \left\langle \bPhi,\bM_{d,l}\bPb^T\bPp\dots\bM_{d,1}\bPb^T\cdot\bPp
            \bD_{PS_1}\bPb^T \pd{\bPp\bw}{x_d}\right\rangle_{\infty} ,
        \end{equation}
        where $l$ is a positive integer and $\bM_{d,i}$,
        $i=1,\dots,l$ are matrices depending on $p_d(\bv)\bPhi$ and the
        internal projection strategy. See Remark \ref{rm:projectionstrategy}
        for details of the upper equations and the internal projection strategy
        $PS_2$. In the following we use $\bM_{d,PS_2}$ to denote
        $\bM_{d,l}\bPb^T\bPp\dots\bM_{d,1}$.
    \item Projection 3: Project \eqref{eq:multiplyvelocity_frame} into
        the space $\bbH^{\weight^{[\bdeta_2]}}_{sub}$ and match the
        coefficients of basis functions $\bPhi$, then obtain the moment
        system
        \begin{equation}\label{eq:momentsystem_frame}
            \begin{aligned}
                &\bPp\bD_{PS_1}\bPb^T\pd{\bPp\bw}{t}
                +\sum_{d=1}^D
                \bPp\bM_{d,PS_2}\bPb^T
                \bPp\bD_{PS_1}\bPb^T\pd{\bPp\bw}{x_d}
                =\bPp\bS,
            \end{aligned}
        \end{equation}
        where $\bS$ is obtained by expansion of the collision part
        $S(f)$, which is not studied in this paper.
\end{enumerate}
\begin{remark}\label{rm:projectionstrategy}
    In the procedure of multiplying velocity, there may be several
    operations involved. As an example we consider $p_d(\bv)=v_d^2$ and we denote the matrix
    $\bM_d$ satisfying $v_d\bPhi=\bM^T_d\bPhi$, then
    $p_d\bPhi=v_d(v_d\bPhi)=v_d\bM_d^T\bPhi=\bM_d^T\bM_d^T\bPhi$.
    Thus, we have two choices for the multiplication with velocity:
    \begin{enumerate}
        \item first compute $v_d\bPhi$ and apply a projection, then perform the
            other multiplication with velocity. This corresponds to $l=2$ and
            $\bM_{d,1}=\bM_{d,2}=\bM_d$.
        \item directly compute $v_d^2\bPhi$. This corresponds to $l=1$ and
            $\bM_{d,1}=\bM_d^2$.
    \end{enumerate}
    If $p_d(\bv)$ is more complex, there are more choices. We call
    each choice an internal projection strategy $PS_2$. Naturally, different
    choices usually yield different moment systems. Here we consider
    the case where $p_d(\bv)$ can be factorized as
    $p_d(\bv)=\prod_{i=1}^lp_d^{(i)}(\bv)$, then $\bM_{d,i}$ satisfies
    $p_d(\bv)\bPhi=\bM_{d,i}^T\bPhi$.
    Similarly, in the procedure of calculating time and space
    derivative, there may be several operations, which result in
    several choices to calculating time and space derivative. We call
    each choice an internal projection strategy $PS_1$, respectively.
\end{remark}
\begin{remark}\label{rm:dependence}
    In the framework, it is assumed that $\bPp$ is commutative with
    the time and space derivative, which means $\bPp$ is independent
    of $\bdeta_1$. Actually, if $\mP$ is an orthogonal projection and
    the basis function is an orthogonal basis, this assumption is
    always valid.
    
    Besides, the ``derivative'' matrix $\bD_{PS_1}$ usually depends on
    the variables $\bw$, e.g. $\bD_{PS_1}=\bD_{PS_1}(\bw)$. After
    projection, the matrix $\bPp\bD_{PS_1}\bPb^T$ must depend only on
    the projected variables $\bPp\bw$ due to the moment closure.
    Actually, we implicitly used the condition:
    $\bPp\bD_{PS_1}\bPb^T=\bPp\bD_{PS_1}(\bPb\bPp\bw)\bPb^T$.
    Similarly,
    $\bPp\bM_{d,PS_2}\bPb^T=\bPp\bM_{d,PS_2}(\bPb\bPp\bw)\bPb^T$ and
    $\bPp\bS=\bPp\bS(\bPb\bPp\bw)$.
\end{remark} 

\subsection{Discussion on the framework}
Actually, the framework in Section \ref{sec:framework_details} almost
provides an algorithm to derive moment systems from kinetic equation.
In this subsection, we dissect the procedure in detail and study the
inputs and properties of the resulting moment system.

\subsubsection{Inputs}
Taking a closer look at the framework in Section
\ref{sec:framework_details}, we find that once the weight function is
given, the weighted polynomial space $\bbH^{\weight^{[\bdeta_2]}}$ is
determined and the ansatz and constraints in the $3$rd and $4$th step of the
framework are also decided. Once the projection operator $\mP$ is
given, all the projections in the $5$th, $7$th and $9$th step are fixed. For
the calculations of the time and space derivative and the
multiplication with velocity, only the internal projection strategy affects the
result. Hence, to derive a moment system based on the framework in
Section \ref{sec:framework_details}, the following information is
needed:
\begin{itemize}
    \item A kinetic equation of the form as in
        \eqref{eq:kineticequation};
    \item Weight function;
    \item Projection operator;
    \item Internal projection strategies $PS_1$ and $PS_2$.
\end{itemize}

As discussed in Section \ref{sec:form_kineticequation}, the form of
the kinetic equation implicates the treatment of the kinetic equation.

The weight function represents some knowledge of the distribution
function. \paperauthor{Grad} used the Maxwellian as the weight function because he
assumed the distribution function is not far away from the
Maxwellian. In \cite{ANRxx}, in order to deal with the anisotropic
distribution function of the Boltzmann equation, \paperauthor{Fan} and
\paperauthor{Li} used a more general Gaussian function as the weight
function. Hence, it is possible to include some prior knowledge of the
distribution function in the weight function, to derive specific
moment systems for some specific questions.

The projection operator largely influences the type of the moment
system. For the conventional Boltzmann equation and the Maxwellian
as the weight function, one projection operator may yield the
regularized version of \Grad 13 moment system (G13) while another one
may yield the regularized version of \Grad 20 moment system (G20).
Even if all the upper three inputs are given, it is possible to obtain
different moment system with different internal projection strategies. So the
internal projection strategy offers some freedom.

As we will see in the later examples and applications, the projection
operators can for example correspond to a truncation or a cut-off
during the computation of the moment system. This will be most obvious
in case of HME and QBME, which are very similar in this new framework.
The operator projection framework thus also yields a mathematically
precise method to describe the procedures of these different
approaches in a unified way.

Summarized, the form of the kinetic equation implicates the treatment
of the kinetic equation. The weight function represents some knowledge
of the distribution function and allows us to include a-priori
information of the distribution function in the moment system.  The
projection operator and internal projection strategy determine which type of
moment system we need and leave us some freedom for the moment system.
Once the four inputs are given, the moment system can be mechanically
derived following the framework.

\subsubsection{Pragmatic viewpoint}
\label{sec:framework_view}
As discussed in the last part of Section \ref{sec:hme1d}, the
internal projection strategy vanishes if we do not apply any projection in the
framework, which is identical to setting $\bPb=\bPp=\identity$.
The resulting moment system then reads
\begin{equation}\label{eq:momentsystem_withoutprojection}
    \bD\pd{\bw}{t}+\sum_{d=1}^D\bM_{d}\bD\pd{\bw}{x_d}=\bS.
\end{equation}
Actually, to derive \eqref{eq:momentsystem_frame}, we can first
neglect the projection operators and obtain
\eqref{eq:momentsystem_withoutprojection}, then afterwards perform the
projections, which can be treated as using $\bPp\bw$ and $\bPp\bS$ to
take the place of $\bw$ and $\bS$, respectively, and use
$\bPp\bD\bPb^T$ and $\bPp\bM_d\bPb^T$ to take the place of $\bD$ and
$\bM_d$. This allows us to choose the weight function first, and then
obtain the moment system containing infinite equations, and finally to
determine the projection.

As emphasized in Section \ref{sec:hme1d}, we note that it is essential
to treat time and space derivative in the same way, which corresponds
to the same internal projection strategy $PS_1$ for time and space derivative.
This is essential for the hyperbolicity of the moment system. Using
the same internal projection strategy $PS_2$ for $\bM_{d,i}$ for different
directions $x_d$ is also obligatory, which corresponds to the
rotational invariance of the resulting moment system. Precisely, the
resulting moment system is always Galilean invariant, since the
subspace is admissible.

\subsubsection{Hyperbolicity of the reduced models}
According to the definition of hyperbolicity, the moment system
\eqref{eq:momentsystem_frame} is hyperbolic if
\begin{enumerate}
    \item $\bPp\bD_{PS_1}\bPb^T$, is invertible;
    \item any linear combination of $\bPp\bM_{d,PS_2}\bPb^T$ is
        diagonalizable with real eigenvalues.
\end{enumerate}

To study the matrix $\bPp\bM_{d,PS_2}\bPb^T$, we denote
$\{\tilde{\varphi}_0,\tilde{\varphi}_1,\dots,\tilde{\varphi}_{N-1}\}$
orthonormal basis of the $N$-dimensional space $\bbH_{sub}^{\weight}$,
satisfying
$
(\tilde{\varphi}_i,\tilde{\varphi}_j)_{\weight}=\delta_{i,j},\;
    i,j=0,\dots,N-1,
$
and denote $\{\tilde{\phi}_0,\dots,\tilde{\phi}_n,\dots\}$ as
orthonormal basis of $\bbH^{\weight}$ with
$\tilde{\phi}_i=\tilde{\varphi}_i, i=0,\dots,N-1$ and
$(\tilde{\phi}_i,\tilde{\phi}_j)_{\weight}=\delta_{i,j}, i,j\in\bbN$,
where $\tilde{\phi}$ is dependent on $\bdeta_2$. Then there exists a
non-singular matrix $\bQ$ such that
$\boldsymbol{{\varphi}}=\bQ^T\boldsymbol{\tilde{\varphi}}$.  In the
new basis, we denote $\tbPp$, $\tbPb$, $\tilde{\bD}_{PS_1}$,
$\tilde{\bM}_{d,PS_2}$, $\tilde{\bM}_{d,k}, k=1,\dots,l$ and
$\tilde{\bw}$ with the same definitions as the symbols without the
$\tilde{\dot}$. Then the resulting moment system can be written as
\begin{equation}
    \bQ \bPp\bD_{PS_1}\bPb^T\pd{\bPp\bw}{t} +\sum_{d=1}^D
    \tbPp\tilde{\bM}_{d,PS_2}\tbPb^T \bQ\bPp\bD_{PS_1}\bPb^T\pd{\bPp\bw}{x_d}
    =\bQ\bPp\bS.
\end{equation}
Hence, we have
\begin{equation}\label{eq:tildebM}
    \bQ^{-1}\tbPp\tilde{\bM}_{d,PS_1}\tbPb\bQ = \bPp\bM_{d,PS_2}\bPb.
\end{equation}
Since $\tilde{\bM}_{d,k}$, $k=1,\dots,l$ is defined by
$p_d^{(k)}(\bv)\tilde{\bPhi}=\tilde{\bM}_{d,k}\tilde{\bPhi}$, and
$\tilde{\bPhi}$ is an orthonormal basis, we have
\begin{equation}\label{eq:tildebMd}
    \tilde{\bM}_{d,k}=\left(
    \left(p_d^{(k)}(\bv)\tilde{\phi}_i,\tilde{\phi}_j\right)_{\weight}
    \right), \quad k=1,\dots,l.
\end{equation}
With \eqref{eq:tildebM} and \eqref{eq:tildebMd}, we immediately get
the following criterion on the real diagonalizability of
$\bPp\bM_{d,PS_2}\bPb^T$.
\begin{theorem}\label{thm:diagonalizable}
    If the projection operator $\mP$ is an orthogonal projection, and
    $p_d^{(k)}(\bv)$, $k=1,\dots,l$ satisfy
    $p_d^{(k)}(\bv)=p_d^{(l+1-k)}(\bv)$, then any linear combination
    of $\bPp\bM_{d,PS_2}\bPb^T$ is diagonalizable with real
    eigenvalues.
\end{theorem}
\begin{proof}
    As $\mP$ is an orthogonal projection, we have
    $\tbPp=\tbPb=\boldsymbol{\mathrm{T}}$.
    Since
    \[
        \left(p_d^{(k)}(\bv)\tilde{\phi}_i,\tilde{\phi}_j\right)_{\weight}
        =
        \left(p_d^{(k)}(\bv)\tilde{\phi}_j, \tilde{\phi}_i\right)_{\weight},
        \quad k=1,\dots,l, d=1,\dots,D,
    \]
    $\tilde{\bM}_{d,k}$ and $\bPp\tilde{\bM}_{d,k}\bPb^T$ are
    symmetric matrices.
    Due to $p_d^{(k)}(\bv)=p_d^{(l+1-k)}(\bv)$, we have
    $\tilde{\bM}_{d,k}=\tilde{\bM}_{d,l+1-k}$, and further
    $\tbPp\tilde{\bM}_{d,PS_2}\tbPb^T=\tbPp\tilde{\bM}_{d,l}\tbPb\dots
    \tbPp\tilde{\bM}_{d,1}\tbPb$, $d=1,\dots,D$ are symmetric matrices.
    Hence, any linear combination of $\tbPp\tilde{\bM}_{d,PS_2}\tbPb^T$
    is diagonalizable with real eigenvalues. \eqref{eq:tildebM}
    indicates the conclusion of the theorem is valid.
\end{proof}

In practice, to derive moment equations, we most often use an
orthogonal projection since it corresponds to the ``cut-off''. Thus
the condition on the projection is almost satisfied. For almost all
kinetic equations, $p_d(\bv)$ is a linear polynomial. Even for some
complex $p_d(\bv)$, using some complex internal projection strategy is not
usual.  Hence, the condition on $p_d(\bv)$ is easy to fulfill. So the
model \eqref{eq:momentsystem_frame} is globally hyperbolic for most
situations, only if the matrix $\bPp\bD_{PS_1}\bPb^T$ is invertible.

Next we consider the matrix $\bPp\bD_{PS_1}\bPb^T$. In this framework,
$\bPp\bw$ can be seen as the parameters to construct a distribution
function $\mP f(\bw; \bxi)$ in $\bbH_{sub}^{\weight}$ to approximate
the solution of the kinetic equation. Generally, it is not permitted
that two different $\bw$ correspond to one distribution function or
one operator $\mL(\pd{}{s};\dots)$, i.e.
\[
    \begin{aligned}
        \bw^0\neq\bw^1\Longrightarrow &\mP f(\bw^0;\bxi)\neq\mP
        f(\bw^1;\bxi),\;~
        \bdeta_1(\bw^0)\neq\bdeta_1(\bw^1),\;~\\
        &\mL(\pd{}{s}; \mP f(\bw^0;\bxi), \bdeta_1(\bw^0),\bv)\neq
        \mL(\pd{}{s}; \mP f(\bw^1;\bxi), \bdeta_1(\bw^1),\bv).
    \end{aligned}
\]
Hence, if the operator $\mL(\pd{}{s};\dots)$ and the weight function
$\weight$ are not singular for some $\bw$, $\bPp\bD_{PS_1}\bPb^T$ is
general invertible.

Before we end this section, we would like to point out that the
framework proposed in this section provides a general model reduction
strategy from kinetic equation to moment equations. The framework is
so concise that we need only routine calculations to obtain a usually
globally hyperbolic moment system. But we also need to point out
whether the moment system is easy to implement or not usually depends
on whether the coefficients of the system are explicit or tractable,
which are significantly up to the ansatz. We will give several
examples, e.g. Section \ref{sec:hme}, \ref{sec:ahme},
\ref{sec:grad13}, \ref{sec:QBME} to show the coefficients of the
moment system are usually explicit and tractable, while the example
Levermore's maximum entropy in Section \ref{sec:Levermore} shows an
opposite side.

\section{Previous Models}
\label{sec:phmsf}
An advantage of the framework is its applicability. Almost all the
traditional moment systems can be derived from the framework. In this
section, we will give several examples of moment systems for the
Boltzmann equation and the radiative transfer equation derived using
the operator projection framework, before we also show an example with
varying projection operators.

\subsection{Hyperbolic moment equations}
\label{sec:hme}
In Section \ref{sec:mmbe}, \Grad moment system for the 1D
Boltzmann equation is studied in detail, and the globally hyperbolic
regularization, proposed in \cite{Fan}, is investigated. Now, we study
the multi-dimensional case. \Grad
moment system of arbitrary order is first proposed in \cite{NRxx}, and in \cite{Fan_new}
the authors investigated the hyperbolicity of it and concluded that
the moment system with order greater than $3$ is not globally
hyperbolic. A globally hyperbolic regularization for it is proposed in
that paper, and here we call the resulting moment system the hyperbolic
moment equations(HME).

For HME, the kinetic equation is the conventional Boltzmann
equation \eqref{eq:boltzmann}, i.e.
\[
    \bdeta_1=\emptyset,\quad
    \bv(\bxi)=\bxi,\quad p_d(\bv)=v_d,\quad
    \mL\left( \pd{}{s}; f,\bdeta_1,\bv \right)=\pd{f}{s}, \; s=t,x_d, d=1,\dots,D.
\]
The weight function is a scaled Maxwellian
\[
\weight^{[\bu,\theta]}=\frac{1}{\sqrt{2\pi\theta}}\exp\left(
-\frac{|\bxi-\bu|^2}{2\theta}
\right).
\]
Then the orthogonal weighted polynomials are defined by
\[
    \mbH_{\alpha}(\bxi)=(-1)^{|\alpha|}\od{^\alpha}{\bxi^{\alpha}}\weight^{[\bu,\theta]},
    \quad \alpha\in\bbN^D,\quad |\alpha|=\sum_{d=1}^D\alpha_d,
\]
which form a basis function of $\bbH^{\weight^{[\bu,\theta]}}$.
We have $\bdeta=\{\bu,\theta\}$, and some calculations yield the constrain
\[
    f_{e_i}=0, \quad i=1,\dots,D,\quad
    \sum_{d=1}^Df_{2e_d}=0.
\]
Hence, we use $u_i$ to replace $f_{e_i}$ and $\theta/2$ to replace
$f_{2e_1}$ in $\bdf$, then set the resulting vector as $\bw$. For
convenience, we denote the consecutive number of $f_{\alpha}$ in $\bdf$ as
$\mN(\alpha)$.

We choose a positive integer $M\geq 2$, the subspace is then defined as
$\bbH_{sub}^{\weight^{[\bu,\theta]}}= \rmspan\left\langle
\left\{\mbH_{\alpha}(\bxi)\right\}_{|\alpha|\leq M} \right\rangle $,
which is an admissible subspace. The projection operator is
chosen as the orthogonal projection, i.e.
$\bPb=\bPp=\boldsymbol{\mathrm{T}}$.  Since $p_d(\bv)$ is a linear
polynomial and $\mL_s$ is only a simple derivative, the projection
strategy vanishes.

With these inputs, we start to derive the moment system. Since
\[
    \begin{aligned}
    \mL\left( \pd{}{s}; f,\bdeta_1,\bv \right)
    &=\sum_{\alpha\in\bbN^D}\mbH_{\alpha}\left(
    \pd{f_{\alpha}}{s}+\sum_{d=1}^Df_{\alpha-e_d}\pd{u_d}{s}
    +\frac{1}{2}\pd{\theta}{s}\sum_{d=1}^Df_{\alpha-2e_d}
    \right),
    \end{aligned}
\]
the matrix $\bD=(d_{ij})$ satisfies
\[
    \begin{aligned}
    &d_{\mN(\alpha),\mN(\alpha)}=1,\quad
    d_{\mN(\alpha), \mN(e_d)}=f_{\alpha-e_d},\quad
    d_{\mN(\alpha), \mN(2e_1)}=\sum_{d=1}^Df_{\alpha-2e_d},\quad
    |\alpha|\neq 1, \text{ and } \alpha\neq 2e_1;\\
    &d_{\mN(e_d),\mN(e_d)}=\rho,\quad
    d_{\mN(2e_1),\mN(2e_1)}=\rho,\quad
    d_{\mN(2e_1),\mN(2e_i)}=-1, \quad d=1,\dots,D, i=2,\dots,D,
    \end{aligned}
\]
and all entries not defined above are zeros. It is easy to observe
that $\bD$ is a block lower triangular matrix, and only the diagonal
block corresponding to rows and
columns from $\mN(2e_1)$ to $\mN(2e_D)$ is a big block, the others are all $1\times1$ blocks and the
entry of the block is nonzero. Hence, we just need to study the big
block, and denote it by $\bD_{\theta}$. For convenience, we just study
the case $D=2$, and it is easy to extend it to the general case. Then
\[
    \bD_{\theta}=\begin{pmatrix}
        \rho    &   0   &   -1\\
        0   &   1   &   0\\
        \rho    &   0   &   1
    \end{pmatrix},\qquad
    \mathrm{det}(\bD_{\theta})=2\rho\neq0,
\]
so the matrix $\bD$ is invertible.

The property of Hermite polynomials give
\[
    \xi_d\mbH_{\alpha}=\theta\mbH_{\alpha+e_d}+u_d\mbH_{\alpha}+\alpha_d\mbH_{\alpha-e_d},
\]
which indicates the form of the matrix $\bM_d$. Since the projection
$\mP$ is an orthogonal projection and $p_d(\bv)$ is a linear polynomial,
Theorem \ref{thm:diagonalizable} indicates the system
\[
    \bPp\bD\bPb^T\pd{\bPp\bw}{t}+
    \sum_{d=1}^D\bPp\bM_d\bPb^T\bPp\bD\bPb^T\pd{\bPp\bw}{x_d}=\bPp\bS.
\]
is hyperbolic. We point out that if we do not perform the projection before
multiplying with the velocity, the resulting moment system turns into
\[
    \bPp\bD\bPb^T\pd{\bPp\bw}{t}+
    \sum_{d=1}^D\bPp\bM_d\bD\bPb^T\pd{\bPp\bw}{x_d}=\bPp\bS,
\]
which is \Grad moment system in \cite{NRxx}.

\subsection{Anisotropic hyperbolic moment equations}
\label{sec:ahme}
HME uses one temperature in the weight function and treats different
directions in the same way. For some anisotropic distribution
functions, for example $f=\frac{\rho}{a(\pi)^{3/2}}\exp\left(
-\frac{\xi_1^2}{a^2}-\xi_2^2-\xi_3^2
\right)$, where $a$ is positive constant, if $a$ is far from $1$, HME
cannot capture this well or even fails to work. In \cite{ANRxx},
\paperauthor{Fan} and \paperauthor{Li} use a Gaussian rather than a Maxwellian as the
weight function and derive an anisotropic hyperbolic moment
equations(AHME). Next, we give a concise derivation of it in our newly proposed framework.

The main difference of AHME from HME is its weight function. Here we use a
Gaussian
\[
    \weight^{[\bu,\Theta]}(\bxi)=\frac{\rho}{\sqrt{\mathrm{det}(2\pi\Theta)}}\exp\left(
    -\frac{1}{2}(\bxi-\bu)^T\Theta^{-1}(\bxi-\bu)
     \right),
\]
where $\Theta=(\theta_{ij})_{D\times D}$, and
$\theta_{ij}=p_{ij}/\rho$. The definition of $p_{ij}$ indicates the
matrix $\Theta$ is positive definite. With the weight function, we
define the generalized Hermite polynomials
\[
    \mHT(\bxi)=(-1)^{|\alpha|}\od{^{\alpha}}{\bx^{\alpha}}\weight^{[\bu,\Theta]},
    \quad \alpha\in\bbN^D,
\]
which are basis functions of $\bbH^{\weight^{[\bu,\Theta]}}$ and
$\bdeta=\{\bu,\Theta\}$. Some calculations yield the constraints
\[
    f_{e_i}=0,\quad f_{e_i+e_j}=0,\quad i,j=1,\dots,D,
\]
so we replace $f_{e_i}$ by $u_i$ and $f_{e_i+e_j}$ by
$\theta_{ij}/(1+\delta_{ij})$ in $\bdf$ and let $\bw$ be the resulting
vector.

We choose an positive integer $M\geq2$, the subspace is then defined as
$\bbH_{sub}^{\weight^{[\bu,\Theta]}}= \rmspan \left\langle
\left\{\mHT_{\alpha}(\bxi)\right\}_{|\alpha|\leq M}\right\rangle $,
which is an admissible subspace. The projection operator is
chosen as the orthogonal projection, and the quasi-orthogonal
property, i.e.  $(\mHT_{\alpha},\mHT_{\beta})_{\weight^{[\bu,\Theta]}}
=\mathrm{Const}_{\alpha}\prod_{d=1}^D\delta_{\alpha_d,\beta_d}$,
indicates $\bPb=\bPp=\boldsymbol{\mathrm{T}}$. As for HME, the projection
strategy vanishes.

With these inputs, the moment system can be derived as follows. Since
\[
    \mL\left( \pd{}{s}; f,\bdeta_1,\bv \right)
    =\sum_{\alpha\in\bbN^D}\mHT_{\alpha}\left(
    \pd{f_{\alpha}}{s}+\sum_{i=1}^Df_{\alpha-e_i}\pd{u_i}{s}
    +\sum_{i,j=1}^D\frac{f_{\alpha-e_i-e_j}}{2}\pd{\theta_{ij}}{s}
    \right)
\]
the matrix $\bD=(d_{ij})$ satisfies, for $1\leq i\leq j\leq D$,
\[
    \begin{aligned}
        &d_{\mN(\alpha),\mN(\alpha)}=1,\quad
        d_{\mN(\alpha),\mN(e_i)}=f_{\alpha-e_i},\quad
        d_{\mN(\alpha),\mN(e_i+e_j)}=f_{\alpha-e_i-e_j},\quad
        |\alpha|\neq 1,2\\
        &d_{\mN(e_i),\mN(e_i)}=\rho,\quad
        d_{\mN(e_i+e_j),\mN(e_i+e_j)}=\rho,
    \end{aligned}
\]
and all entries, not defined above, are zeros. It is easy to observe
that $\bD$ is a low-triangular matrix and the diagonal entries are all
non-zero, hence $\bD$ is invertible.

The property of generalized Hermite polynomials give
\[
    \xi_d\mHT_{\alpha}=\sum_{j=1}^D\theta_{jd}\mHT_{\alpha+e_j}+u_d\mHT_{\alpha}
    +\alpha_d\mHT_{\alpha-e_d},
\]
which indicates the form of the matrix $\bM_d$.
Since the projection $\mP$ is an orthogonal projection and $p_d(\bv)$ is
a linear polynomial, Theorem \ref{thm:diagonalizable} indicates the
system
\[
    \bPp\bD\bPb^T\pd{\bPp\bw}{t}+
    \sum_{d=1}^D\bPp\bM_d\bPb^T\bPp\bD\bPb^T\pd{\bPp\bw}{x_d}=\bPp\bS.
\]
is globally hyperbolic.

Particularly, the moment system with $M=2,D=3$ is the 10 moment system
with Gaussian closure \cite{Gaussian10}.

\subsection{G13 moment system with hyperbolic regularization}
\label{sec:grad13}
Among all of \Grad moment systems, the G13 moment system drew most
attention of researchers. However, the system suffers a serious
problem with its hyperbolicity. In \cite{Grad13toR13}, it is reported
that the hyperbolicity of it cannot be ensured even around the
Maxwellian. Recently, in \cite{framework}, the authors proposed a
hyperbolic regularization for it. Now we put it in the framework in
detail to help readers to understand the framework. Here we need to
point out that this subsection is similar as Section 4.1.1 in
\cite{framework} since the procedure of the derivative of the moment
system is same.

The kinetic equation and the weight function are the same as those of
HME with $D=3$, and the only difference is the projection. Since in
Section \ref{sec:hme} the moment system with infinite equations has
been derived, based on the idea in Section \ref{sec:framework_view},
we just need to give the projection. The symbols $\bw$, $\bD$ and
$\bM_d$ have the same definition as that in Section \ref{sec:hme}.

For the 13 moment system, only $\rho, u_i, p_{ij},q_i$, $i,j=1,\dots,3$
are taken into account, hence the subspace is
$\bbH^{\weight^{[\bu,\theta]}}_{sub}= \rmspan\left\langle
\weight^{[\bu,\theta]} \left\{1,\xi_i,\xi_i\xi_j,|\bxi|^2\xi_i\right\}
\right\rangle $.  We choose the basis of
$\bbH^{\weight^{[\bu,\theta]}}_{sub}$ as
$\left\{\mbH_{\alpha}\right\}_{|\alpha|\leq2} \bigcup
\left\{\sum_{d=1}^D\mbH_{e_i+2e_d}, i=1,\dots,3\right\}$, then the
matrix $\bPb=(p_{b,ij})_{13\times\infty}$ is
\[
    \begin{aligned}
        p_{b,i,i}=1, i=1,\dots,10,\quad
        &p_{b, 11,\mN(3e_1)}=1,\quad
        p_{b, 11,\mN(e_1+2e_2)}=1,\quad
        p_{b, 11,\mN(e_1+2e_3)}=1,\\
        &p_{b, 12,\mN(3e_2)}=1,\quad
        p_{b, 12,\mN(e_2+2e_1)}=1,\quad
        p_{b, 12,\mN(e_2+2e_3)}=1,\\
        &p_{b, 13,\mN(3e_3)}=1,\quad
        p_{b, 13,\mN(e_3+2e_1)}=1,\quad
        p_{b, 13,\mN(e_3+2e_2)}=1,
    \end{aligned}
\]
where $\mN(\alpha)$ is the same as in the definition for HME, and all
entries, not defined above, are zero.

The orthogonal projection is used for the 13 moment system, so some
calculations based on \eqref{eq:bpp} give the matrix
$\bPp=(p_{p,ij})_{13\times\infty}$ as
\[
    \begin{aligned}
        p_{p,i,i}=1, i=1,\dots,10,\quad
        p_{p,11,\mN(3e_1)}=\frac{3}{5},\quad
        p_{p,11,\mN(e_1+2e_2)}=\frac{1}{5},\quad
        p_{p,11,\mN(e_1+2e_3)}=\frac{1}{5},\\
        p_{p,12,\mN(3e_2)}=\frac{3}{5},\quad
        p_{p,12,\mN(e_2+2e_1)}=\frac{1}{5},\quad
        p_{p,12,\mN(e_2+2e_3)}=\frac{1}{5},\\
        p_{p,13,\mN(3e_3)}=\frac{3}{5},\quad
        p_{p,13,\mN(e_3+2e_1)}=\frac{1}{5},\quad
        p_{p,13,\mN(e_3+2e_2)}=\frac{1}{5},\\
    \end{aligned}
\]
and all entries, not defined above, are zero again.
Easy to check, we have $\bPp\bw=\bw_{13}$, where
$\bw_{13}=(\rho,u_1,u_2,u_3,\theta/2,f_{e_1+e_2},f_{e_1+e_3},f_{2e_2},f_{e_2+e_3},f_{2e_3},
q_1/5,q_2/5,q_3/5)^T$.
Remark \ref{rm:dependence} indicates that $\bw$ is replaced by
$\bPb\bPp\bw=\bPb\bw_{13}$, which yields
\[
    f_{e_i+e_j+e_k}=\frac{1}{5}(\delta_{ij}q_k+\delta_{ik}q_j+\delta_{jk}q_i),
    \quad f_{\alpha}=0, |\alpha|\geq4.
\]

First, the ansatz is 
\[
    \mP f = \sum_{|\alpha|\leq2}f_{\alpha}\mbH_{\alpha}(\bxi) + 
    \frac{1}{5}\sum_{i,j=1}^3q_i\mbH_{e_i+2e_j}(\bxi),
\]
with $f_{e_i}=0, i=1,2,3$ and $\sum_{i=1}^3f_{2e_i}=0$.
Let 
\[
    \sigma_{ij}=\int_{\bbR^3}(\xi_i-u_i)(\xi_j-u_j)f\dd\bxi =
    (1+\delta_{ij})f_{e_i+e_j}.
\]
Then the time and space derivative can be calculated directly as
\[
    \begin{aligned}
        \pd{\mP f}{s} &=
        \pd{\rho}{s}\mbH_0(\bxi) 
        +\sum_{d=1}^3\rho\pd{u_d}{s}\mbH_{e_d}(\bxi)
        +\frac{1}{2}\rho\pd{\theta}{s}\sum_{d=1}^3\mbH_{2e_d}(\bxi)
        +\frac{1}{2}\sum_{i,j=1}^3\pd{\sigma_{ij}}{s}\mbH_{e_i+e_j}(\bxi)\\
        &\quad +\frac{1}{5}\sum_{i,j=1}^3\pd{q_i}{s}\mbH_{e_i+2e_j}(\bxi) 
        +\underline{\underline{\sum_{i,j,d=1}^3\frac{\sigma_{ij}}{2}\pd{u_d}{s}\mbH_{e_i+e_j+e_d}(\bxi)}}
        +\underline{\frac{1}{4}\pd{\theta}{s}\sum_{i,j,d=1}^3\sigma_{ij}\mbH_{e_i+e_j+2e_d}(\bxi)}\\
        &\quad +\underline{\frac{1}{5}\sum_{i,j,d=1}^3q_i\pd{u_d}{s}\mbH_{e_i+2e_j+e_d}(\bxi)}
        +\underline{\frac{1}{10}\pd{\theta}{s}\sum_{i,j,d=1}^3q_i\mbH_{e_i+2e_j+2e_d}(\bxi)}
        \\
        &=\left\langle \bmH,
        \bD\bPb^T\pd{\bPp\bw}{s}\right\rangle_{\infty},
        \qquad s=t,x_k, k=1,2,3.
    \end{aligned}
\]
Projecting $\pd{\mP f}{s}$ into the subspace
$\bbH^{\weight^{[\bu,\theta]}}_{sub}$ is in fact discarding all the
underlined terms and revising the double underlined terms as 
\[
        \mP \sum_{i,j,d=1}^3\frac{\sigma_{ij}}{2}\pd{u_d}{s}\mbH_{e_i+e_j+e_d}(\bxi)
        =\frac{1}{5}\sum_{i,j,d=1}^3
        \sigma_{ij}\pd{u_j}{s}\mbH_{e_i+2e_d}(\bxi).
\]
Till now, we have calculated $\mP\pd{\mP
f}{s}=\left\langle\bPb\bmH,\bPp\bD\bPb^T\pd{\bPp\bw}{s}\right\rangle_{13}$.
For the convection term, Grad directly multiplied $\pd{\mP f}{x_k}$ by
velocity $x_k$ while in our framework we multiplied $\mP\pd{\mP
f}{x_d}$ by velocity $x_k$. Direct calculations give the expression of
$(\xi_k-u_k)\pd{\mP f}{x_k}$ and $(\xi_k-u_k)\mP\pd{\mP f}{x_k}$ as
\[
    \begin{aligned}
        &\rho\pd{u_k}{x_k}\mbH_0(\bxi)
        + \pd{\rho\theta}{x_k}\mbH_{e_k}(\bxi)+\sum_{i=1}^3\pd{\sigma_{ik}}{x_k}\mbH_{e_i}(\bxi)
        +\sum_{d=1}^3\left( \rho\theta\pd{u_d}{x_k}
        +\frac{2}{5}\pd{q_d}{x_k} \right)\mbH_{e_k+e_d}(\bxi)\\
        &+ \sum_{j=1}^3\frac{1}{5}\pd{q_k}{x_k}\mbH_{2e_j}(\bxi)
        + \sum_{d=1}^3\frac{\rho\theta}{2}\pd{\theta}{x_k}\mbH_{2e_d+e_k}(\bxi)
        +\underline{\underline{\sum_{i,j=1}^3\frac{\theta}{2}\pd{\sigma_{ij}}{x_k}\mbH_{e_i+e_j+e_k}(\bxi)}}\\
        &+\left\{  \begin{array}{ll}
            \text{C1:} & 
            \begin{aligned}
                & \sum\limits_{i,j=1}^3\left(
                \sigma_{ki}\pd{u_j}{x_k}+\frac{1}{2}\sigma_{ij}\pd{u_k}{x_k}
                \right)\mbH_{e_i+e_j}(\bxi)
                +\frac{1}{2}\pd{\theta}{x_k}\sum\limits_{i,j=1}^3(\sigma_{kj}\mbH_{e_j+2e_i}+\underline{\sigma_{ij}\mbH_{e_i+e_j+e_k}})\\
                &\quad + \frac{1}{5}\sum_{i,j=1}^3\left(\left(
                q_k\pd{u_i}{x_k}+q_i\pd{u_k}{x_k}
                \right)\mbH_{e_i+2e_j}(\bxi) +
                \underline{2q_i\pd{u_j}{x_k}\mbH_{e_i+e_j+e_d}(\bxi)}\right)
                +h.o.t.
            \end{aligned}\\
            \text{C2:} & \sum\limits_{i,j=1}^3\left( 
            \frac{1}{5}\sigma_{kj}\pd{u_j}{x_k}\mbH_{2e_i}(\bxi)+\frac{2}{5}\sigma_{ij}\pd{u_j}{x_k}\mbH_{e_i+e_k}(\bxi)
            \right)+h.o.t.
        \end{array}\right.,
    \end{aligned}
\]
where $h.o.t.$ denotes by the terms with $\mbH_{\alpha}(\bxi)$,
$|\alpha|>3$, and C1 and C2 correspond to $(\xi_k-u_k)\pd{\mP f}{x_k}$
and $(\xi_k-u_k)\mP\pd{\mP f}{x_k}$, respectively.  These calculations
give $(\xi_k-u_i)\pd{\mP f}{x_k} = \left\langle(\xi_k-u_k)\bmH,
\bD\bPb^T\pd{\bPp\bw}{x_k}\right\rangle_{\infty} =\left\langle\bmH,
(\bM_k-u_k\identity)\bD\bPb^T\pd{\bPp\bw}{x_k}\right\rangle_{\infty}$ and
$(\xi_k-u_i)\mP\pd{\mP f}{x_k} = \left\langle(\xi_k-u_k)\bmH,
\bPb^T\bPp\bD\bPb^T\pd{\bPp\bw}{x_k}\right\rangle_{\infty} =\left\langle\bmH,
(\bM_k-u_k\identity)\bPb^T\bPp\bD\bPb^T\pd{\bPp\bw}{x_k}\right\rangle_{\infty}$.
Projecting $(\xi_k-u_k)\pd{\mP f}{x_k}$ and $(\xi_k-u_k)\mP\pd{\mP
f}{x_k}$ into the subspace $\bbH^{\weight^{[\bu,\theta]}}_{sub}$ is in
fact discarding $h.o.t.$ terms and revising the double
underlined terms as 
\[
    \mP \sum_{i,j=1}^3\frac{\theta}{2}\pd{\sigma_{ij}}{x_k}\mbH_{e_i+e_j+e_k}(\bxi)
    =\sum_{i,j=1}^3\frac{\theta}{5}\pd{\sigma_{ik}}{x_k}\mbH_{e_i+2e_j}(\bxi),
\]
and revising the underlined terms as
\[
    \begin{aligned}
        \mP \frac{1}{2}\pd{\theta}{x_k}\sum\limits_{i,j=1}^3\sigma_{ij}\mbH_{e_i+e_j+e_k}
        &=
        \frac{1}{5}\sum_{i,j=1}^3\sigma_{ki}\pd{\theta}{x_k}\mbH_{e_i+2e_j}(\bxi),\\
        \mP \sum_{i,j=1}^3 \frac{2}{5}q_i\pd{u_j}{x_k}\mbH_{e_i+e_j+e_d}(\bxi)
        &=\frac{2}{25}\sum_{i,j=3}^5\left(
        (q_i\pd{u_i}{x_k}+q_i\pd{u_k}{x_k})\mbH_{e_i+2e_j}(\bxi)
        +q_k\pd{u_i}{x_k}\mbH_{e_k+2e_j}(\bxi) \right).
    \end{aligned}
\]
Till now, we finish the convection term and obtain 
$\mP(\xi_k-u_i)\pd{\mP f}{x_k}$ $=$ $ \left\langle\bPb\bmH,
\bPp(\bM_k-u_k\identity)\bD\bPb^T\pd{\bPp\bw}{x_k}\right\rangle_{13}$ and
$\mP(\xi_k-u_i)\mP\pd{\mP f}{x_k} = \left\langle\bPb\bmH,
\bPp(\bM_k-u_k\identity)\bPb^T\bPp\bD\bPb^T\pd{\bPp\bw}{x_k}\right\rangle_{13}$.

Then matching the coefficients of $\bPb\mbH$, we obtain the well-known
\Grad 13 moment system(G13) and hyperbolic regularized 13 moment
system(HR13):
\begin{equation}
    \begin{aligned}
        \od{\rho}{t}&+\sum_{d,k=1}^3\rho\pd{u_k}{x_k}=0,\\
        \rho\od{u_i}{t}&+\pd{\rho\theta}{x_i}+\sum_{k=1}^3\pd{\sigma_{ki}}{x_k}=0, \\
        \frac{3\rho}{2}\od{\theta}{t}&+\sum_{k=1}^3\pd{q_k}{x_k}
        +\sum_{k=1}^3\rho\theta\pd{u_k}{x_k}+\sum_{k,d=1}^3\sigma_{kd}\pd{u_d}{x_k}=0,\\
        \od{\sigma_{ij}}{t}&+2\rho\theta\pd{u_{\langle
        i}}{x_{j\rangle}}+\frac{4}{5}\pd{q_{\langle i}}{x_{j\rangle}} 
        +\left\{ \begin{array}{ll}
            \text{G13:} & \sum_{k=1}^3\left(2\sigma_{k\langle
            i}\pd{u_{j\rangle}}{x_k}+\sigma_{ij}\pd{u_k}{x_k}\right)\\
            \text{HR13:} & \frac{4}{5}\sigma_{k\langle
            i}\pd{u_k}{x_{j\rangle}}
        \end{array} \right. 
        =S(\sigma_{ij}),\\
        \od{q_i}{t} &+\sum_{j=1}^3\sigma_{ij}\pd{u_j}{t}
        +\frac{5\rho\theta}{2}\pd{\theta}{x_i}+\sum_{k=1}^3\theta\pd{\sigma_{ik}}{x_k}\\
        &+\left\{ \begin{array}{ll}
            \text{G13:} &
            \sum_{k=1}^3\left(\frac{7}{2}\sigma_{ki}\pd{\theta}{x_k}
            +\frac{7}{5}q_i\pd{u_k}{x_k}+\frac{7}{5}q_k\pd{u_i}{x_k}+\frac{2}{5}q_k\pd{u_k}{x_i}\right)
            \\
            \text{HR13:} & 0
        \end{array} \right.
        =S(q_i),
    \end{aligned}
\end{equation}
where $\od{\cdot}{t}=\pd{\cdot}{t}+\sum_{k=1}^3u_k\pd{\cdot}{x_k}$ is
the material derivative, and in the governing equation of
$\sigma_{ij}$, the trace-free tensor symbol is used, which is defined
as for a tensor $t_{ij}$, $t_{\langle
ij\rangle}=\frac{1}{2}(t_{ij}+t_{ji})-\sum_{k=1}^3\frac{1}{3}t_{kk}$.

The upper systems can be easily written in the form as 
\[
    \begin{aligned}
        \text{G13:} \quad&\bPp\bD\bPb^T\od{\bw_{13}}{t}+\sum_{d=1}^3\bPp(\bM_d-u_k\identity)\bD\bPb^T\pd{\bw_{13}}{x_d}
    =\bPp\bS,\\
        \text{HR13:} \quad&\bPp\bD\bPb^T\od{\bw_{13}}{t}+\sum_{d=1}^3\bPp(\bM_d-u_k\identity)\bPb^T\bPp\bD\bPb^T\pd{\bw_{13}}{x_d}
    =\bPp\bS,
    \end{aligned}
\]
and HR13 is exactly the regularized 13 moment system in \cite{framework} and is globally
hyperbolic.

\subsection{Maximum entropy closure}
\label{sec:Levermore}
\paperauthor{Levermore} investigated the maximum entropy principle and
provided a moment closure hierarchy for the Boltzmann equation in
\cite{Levermore}. The resulting moment system possesses an entropy and
global hyperbolicity, while it is known for the lack of a simple
analytical expression. Nevertheless, we try to put the moment system
in our framework. For convenience, only the case $D=1$ is studied, but
there is no essential difficulty to extend this to multi-dimensional cases.

For an even and positive integer $M$, Levermore's linear subspace
$\mathbb{M}$ is defined by
$\mathbb{M}=\rmspan\left\langle 1,\xi,\dots,\xi^M\right\rangle $. Based on the maximum
entropy principle, the distribution function is assumed to have the following
form
\[
    \mathcal{M}(\boldsymbol{g})=\exp\left(
    \boldsymbol{g}^T\boldsymbol{\psi}\right),
\]
where $\boldsymbol{g}\in\bbR^{M+1}$ is a vector of some macroscopic parameters
and $\boldsymbol{\psi}=(1,\xi,\dots,\xi^M)^T$. Choose the weight
functions as $\weight^{[\boldsymbol{g}]}=\mathcal{M}(\boldsymbol{g})$,
then using the Schmidt orthogonalization, we can obtain an
orthogonal basis $\phi_i$ of $\bbH^{\weight^{[\boldsymbol{g}]}}$
satisfying $\phi_i/\weight^{[\boldsymbol{g}]}$ is a monic polynomial
with degree $i$, i.e. there exist constants $c_{m,i}(\boldsymbol{g}),
i=0,\dots,m-1$ such that
\[
    \phi_i = \weight^{[\boldsymbol{g}]}
    \left( \xi^m+\sum_{i=0}^{m-1}c_{m,i}\xi^i \right).
\]
If we let $c_{i,i}=1$ and $c_{i,j}=0, j>i$, $i,j=1,\dots,M+1$, then
$\boldsymbol{\varphi} =
\weight^{[\boldsymbol{g}]}\boldsymbol{\mathrm{C}}\boldsymbol{\psi}$, where
$\boldsymbol{\varphi}=(\phi_0,\dots,\phi_M)^T$ and
$\boldsymbol{\mathrm{C}} = (c_{i-1,j-1})_{M+1\times M+1}$. Since
$\boldsymbol{\mathrm{C}}$ is a lower triangular matrix and its diagonal
entries are all zero, $\boldsymbol{\mathrm{C}}$ is invertible.

Set the subspace as $\bbH_{sub}^{\weight^{[\boldsymbol{g}]}}=
\{\weight^{[\boldsymbol{g}]}h|h\in\mathbb{M}\}$ and $\phi_i,
i=0,\dots,M$, as the basis. Furthermore, an orthogonal projection is used, thus
$\bPp=\bPb=\boldsymbol{\mathrm{T}}$.
Since Levermore assumed the distribution function had the form
$\mathcal{M}(\boldsymbol{g})$, we have $\mP f=\mathcal{M}(\boldsymbol{g})$.
So the constraints are
\[
    f_0=1, \quad f_i=0, i=1,\dots,M,
\]
and $\bw$ is set to $\bw(g_0,\dots,g_M,f_{M+1},\dots)^T$. We write $\bw_{M+1}=\bPp\bw=\boldsymbol{g}$.

Now we begin to derive the moment system. The time and space
derivative turns out to be
\[
    \mL_s(\mP f, \bdeta_1,\bxi) =
    \left\langle
        \weight^{[\boldsymbol{g}]}\boldsymbol{\psi},\pd{\boldsymbol{g}}{s}
    \right\rangle_{N}
    =\left\langle \boldsymbol{\varphi},
    \boldsymbol{\mathrm{C}}^{-T}\pd{\boldsymbol{g}}{s}\right\rangle_{N}
    =\left\langle \bPhi, \tilde{\bD}\pd{\bw_{M+1}}{s}\right\rangle_{N} ,
\]
where $\tilde{\bD}=\bD\bPb^T=\bPb^T\boldsymbol{\mathrm{C}}^{-T}$.

Since $\phi_i/\weight^{[\boldsymbol{g}]}$ is an orthogonal
polynomial, there exist three-term recurrence relations, i.e.
there exist $r_{i,j}(\boldsymbol{g}), j=i-1,i,i+1$ such that
\[
    r_{i,i+1}\phi_{i+1}=(\xi-r_{i,i})\phi_i-r_{i,i-1}\phi_{i-1}.
\]
Denote $\bM^T=(m_{ij})$ by $m_{i+1,j+1}=r_{i,j},j=i-1,i,i+1,
i=0,1,\dots$ and $m_{i+1,j+1}=0, j\neq i-1,i,i+1, i=0,1,\dots$,
then
\[
    \bPp\tilde{\bD}\pd{\bw_{M+1}}{t}+
    \bPp\bM\tilde{\bD}\pd{\bw_{M+1}}{x}=\bPp\bS
\]
is Levermore's moment system. Since $\bPp\bM\tilde{\bD}
= \bPp\bM\bPb^T\boldsymbol{\mathrm{C}}^{-T}
= \bPp\bM\bPb^T\bPp\bPb^T\boldsymbol{\mathrm{C}}^{-T}$, the moment system
\[
    \bPp\tilde{\bD}\pd{\bw_{M+1}}{t}+
    \bPp\bM\bPb^T\bPp\tilde{\bD}\pd{\bw_{M+1}}{x}=\bPp\bS,
\]
derived by our framework, is also Levermore's moment system.

\subsection{Quadrature-based moment equations}
\label{sec:QBME}
Different from HME, a new globally hyperbolic regularization for \Grad moment system was proposed by \textsc{Koellermeier} et al.
in \cite{Koellermeier} and \cite{KoellermeierMSc2013}, recently. The
underlying idea of their quadrature-based moment equations (QBME) is
the substitution of the projection method from analytical integration
to quadrature formulas. With the help of a new framework in
\cite{Koellermeier}, it was shown that the emerging system of
equations is in fact hyperbolic and the eigenvalues also correspond to
the Hermite roots. Now we would like to give a concise deduction of
the one-dimensional quadrature-based moment equations in terms of the
proposed framework of this paper.

For QBME, the 1D Boltzmann equation is considered and the kinetic
equation reads
\begin{equation*}
    \begin{aligned}
        &\bdeta_1=(u,\theta),\quad
        v(\xi)=\frac{\xi-u}{\sqrt{\theta}},\quad
        p(v)=u+\sqrt{\theta}v,\quad\\
        &\mL\left( \pd{}{s};f,\bdeta_1,v\right)=\pd{f}{s}-\pd{f}{v}
        \left( \frac{1}{\sqrt{\theta}}\pd{u}{s}
        +\frac{1}{2\theta}v\pd{\theta}{s}
        \right),\; s=t, x,
    \end{aligned}
\end{equation*}
where $f=f(t,x,v)$. The weight function and the orthogonal weighted
polynomials are defined by
\[
    \weight(v)=\frac{1}{\sqrt{2\pi}}\exp\left( -\frac{v^2}{2}
    \right),\quad
    \mh_{k}(v)=(-1)^{k}\od{^k\weight}{v^k},\quad k\in\bbN,
\]
where $\mh_k(v)/\weight(v)$ is the classical Hermite polynomials.
The orthogonal weighted polynomials satisfy the following properties:
\begin{itemize}
    \item Differential relation: $\od{\mh_k(v)}{v}=-\mh_{k+1}(v)$;
    \item Recurrence relation: $\mh_{k+1}(v)=v\mh_k(v)-k\mh_{k-1}(v)$.
\end{itemize}
For convenience, we define the matrix $\bD_v=(d_{ij})$ with
$d_{ij}=\delta_{i,j+1}$ and $\bM_v=(m_{ij})$ with $m_{i,i+1}=i$,
$m_{i+1,i}=1$  and all others entries set to zeros. Then we have
$\od{\boldsymbol{\mh}}{s}=-\bD_v^T\boldsymbol{\mh}$ and
$v\boldsymbol{\mh}=\bM_v^T\boldsymbol{\mh}$, where $\boldsymbol{\mh}=
(\mh_0,\dots,\mh_n,\dots)^T$. Since $\bdeta=\{u,\theta\}$, some
calculations yield the constraints
\[
    f_1 = f_2 = 0.
\]
We choose $\bw$ as $(f_0, u, \theta,f_3,\dots,f_k,\dots)^T$. Choose a
positive integer $M\geq3$, the subspace is then defined as
$\bbH_{sub}^{\weight}
=\rmspan\left\langle \left\{\mh_{k}\right\}_{k\leq
M}\right\rangle$. The projection operator is chosen as the orthogonal
projection, i.e. $\bPb=\bPp=\boldsymbol{\mathrm{T}}$.

For the time and space derivative, we have
\[
    \begin{aligned}
        \mL\left( \pd{}{s}; f,\bdeta_1,v\right)
        &= \left\langle \boldsymbol{\mh}, \pd{\bdf}{s}\right\rangle_{\infty}
        -\left\langle \od{\boldsymbol{\mh}}{v}, \bdf \left(
        \frac{1}{\sqrt{\theta}}\pd{u}{s} +\frac{1}{2\theta}v\pd{\theta}{s}
        \right)\right\rangle_{\infty} \\
        &= \left\langle \boldsymbol{\mh}, \pd{\bdf}{s}\right\rangle_{\infty}
        + \left\langle \boldsymbol{\mh}, \bD_v\bdf
        \frac{1}{\sqrt{\theta}}\pd{u}{s} \right\rangle_{\infty}
        +\left\langle \boldsymbol{\mh},\bM_v\bD_v\bdf \frac{1}{2\theta}\pd{\theta}{s}
        \right\rangle_{\infty}.
    \end{aligned}
\]
For the derivative term, there are two matrices in the last term of the
upper equation. The internal projection strategy $PS_1$ is
\[
    \begin{aligned}
        \mL^{PS_1}\left( \pd{}{s}; f,\bdeta_1,v\right) =
        \left\langle \boldsymbol{\mh}, \bPb^T\pd{\bPp\bdf}{s}\right\rangle_{\infty}
        &+\left\langle \boldsymbol{\mh}, \bD_v\bPb^T\bPp\bdf
        \frac{1}{\sqrt{\theta}}\pd{u}{s} \right\rangle_{\infty} \\
        &+\left\langle \boldsymbol{\mh},\bM_v\bPb^T\bPp\bD_v\bPb^T\bPp\bdf \frac{1}{2\theta}\pd{\theta}{s}
        \right\rangle_{\infty}.
    \end{aligned}
\]
Collecting all the coefficients of $\pd{\bw}{s}$, we obtain
$\bPp\bD_{PS_1}\bPb^T=(d_{ps,i,j})_{M+1,M+1}$ satisfying
\[
    \begin{aligned}
        d_{ps,i,i}&=1,\;~ i=1,4,5,\dots,M+1,\quad
        d_{ps,i,2} = \frac{f_{i-2}}{\sqrt{\theta}},\;~
        i=1,\dots,M+1,\\
        d_{ps,i,3}&=\frac{1}{2\theta}(f_{i-3}+i f_{i-1}),\;~
        i=1,\dots,M, \quad
        d_{ps,i,M+1}=\frac{f_{i-3}}{2\theta}.
    \end{aligned}
\]
It is easy to verify the invertibility of $\bPp\bD_{PS_1}\bPb^T$.
Since $p(v)=u+\sqrt{\theta}v$, the internal projection strategy $PS_2$ vanishes
and $\bM=u\identity +\sqrt{\theta}\bM_v$.  Since the projection $\mP$
is an orthogonal projection and $p_d(\bv)$ is a linear polynomial,
Theorem \ref{thm:diagonalizable} indicates the resulting system
\[
    \bPp\bD_{PS_1}\bPp^T\pd{\bw}{t}+\bPp\bM\bPb^T\bPp\bD_{PS_1}\bPb^T\bPp\pd{\bw}{x}=\bPp\bS
\]
is globally hyperbolic.

The derivation shows that even the QBME with substitution of
exact integration by a suitable quadrature rule can be interpreted as
a certain projection method, where the internal projection strategy $PS_1$ is
particularly important. In fact, the additional projection in $PS_1$
reflects the additional cut-off of higher order terms that is done by
quadrature-based methods automatically during the calculation, see
e.g. the hyperbolicity proof in \cite{Koellermeier}.

Here we point out that if the internal projection strategy $PS_1$ is chosen as
\[
    \begin{aligned}
        \mL^{PS_1}\left( \pd{}{s}; f,\bdeta_1,v\right) =
        \left\langle \boldsymbol{\mh}, \bPb^T\pd{\bPp\bdf}{s}\right\rangle_{\infty}
        &+\left\langle \boldsymbol{\mh}, \bD_v\bPb^T\bPp\bdf
        \frac{1}{\sqrt{\theta}}\pd{u}{s} \right\rangle_{\infty} \\
        &+\left\langle \boldsymbol{\mh},\bM_v\bD_v\bPb^T\bPp\bdf \frac{1}{2\theta}\pd{\theta}{s}
        \right\rangle_{\infty},
    \end{aligned}
\]
i.e. without the additional projection in the last term, the resulting
moment system is the same as the HME moment system \eqref{eq:hme1d} in
Section \ref{sec:hme1d}.

\subsection{Model reduction with alternative projection operators}
Apart from the choice of the equation, the basis functions and the
internal projection strategy, there is also the possibility to use different
projection operators to derive existing and new moment systems.

In the framework proposed in Section \ref{sec:framework} there are
three projections, i.e. projection of the distribution function, the
time and space derivative and the term after multiplying with velocity
into the subspace $\bbH_{sub}^{\weight}$. Different projections
correspond to different steps during the computation, thus it can be
reasonable to use different projection operators in the framework.
Then the resulting moment system can be written as
\begin{equation}\label{eq:momentsystem_frame_new}
    \begin{aligned}
        \bPp^{(2)}\bD_{PS_1}\bPb^T\pd{\bPp^{(1)}\bw}{t}
        +\sum_{d=1}^D
        \bPp^{(3)}\bM_{d,PS_2}\bPb^T
        \bPp^{(2)}\bD_{PS_1}\bPb^T\pd{\bPp^{(1)}\bw}{x_d}
        =\bPp\bS,
    \end{aligned}
\end{equation}
where $\bPp^{(k)}$, $k=1,2,3$ correspond to three projections and
$\bPp$ is some projection for the right side hand, which is not
concerned in this paper. In calculating $\bD_{PS_1}$ and
$\bM_{d,PS_2}$, the projections $\bPp^{(2)}$ and $\bPp^{(3)}$ are
used. The procedure of the framework requires $\bPp^{(1)}$ to be
commutative with time and space derivative, and the conditions in
Theorem \ref{thm:diagonalizable} restrict $\bPp^{(3)}$ to an
orthogonal projection. Based on this idea, it is possible to derive a
different type of moment system for the same inputs of the framework
except for the projection. However, we remark that for the standard
projection, we have $\mP^2=\mP$, but for different projections $\mP_1$
and $\mP_2$, $\mP_2\mP_1$ is usually not equal to $\mP_2$. The
viewpoint in Section \ref{sec:framework_view} may fail to work. In the
procedure of the framework, more attention should be paid on the
derivation.

As an example for the derivation of an existing system, we consider
the 1D QBME, described in Section \ref{sec:QBME}.

If we choose $\bPp^{(1)} = \boldsymbol{\mathrm{T}}$, then compute the
time and space derivative, we get
\[
    \mL^{PS_1}\left( \pd{}{s};\mP f,\emptyset, \xi \right)
    =\left\langle \bPhi,\bD^d\bPb^T\pd{\bPp\bw}{s}\right\rangle_{\infty},
\]
where $\bD^d\bPb^T$ is $\bD\bPb^T$ with $f_k=0$ for $k>M$, and $\mL$,
$\mP$, $f$, $\bw$, $\bPhi$ and $\bD$ have the same definition as that
in \ref{sec:hme}.  We choose the second projection $\bPp^{(2)}$ as
\[
    \bPp^{(2)} = \boldsymbol{\mathrm{T}} - \frac{M+1}{\theta}
    \boldsymbol{\mathrm{E}}_{M+1,M+3},
\]
where $\boldsymbol{\mathrm{E}}_{i,j}$ is a matrix with only the
$i,j$-entry is one and the others are all zero, and the size of it
depends on the context.
Choosing the third projection $\bPp^{(3)}=\boldsymbol{\mathrm{T}}$,
the resulting moment equations
\[
    \bPp^{(2)}\bD^d\bPb^T\pd{\bPp^{(1)}\bw}{t} +
    \bPp^{(3)}\bM\bPb^T\bPp^{(2)}\bD^d\bPb^T\pd{\bPp^{(1)}\bw}{x}=\bPp\bS.
\]
are the quadrature-based moment equations \cite{Koellermeier}, and the
same as those in Section \ref{sec:QBME}.

From the point of view of using different projections in the
framework, we can see the difference between HME and QBME for 1D,
which is only the use of a different projection operator. It shows,
that the methods are in fact closely related and belong to the same
type of projection method. The same procedure is unfortunately not
possible in the multi-dimensional case, as the basis functions do not
match.

This treatment also offers some flexibility for the hyperbolicity. The
eigenvalues of the coefficient matrix of the system all depend on the
matrix $\bPp^{(3)}\bM_{d,PS_2}\bPb^T$, and the only constraint on the
matrix $\bPp^{(2)}\bD_{PS_1}\bPb^T$ is the invertibility. Hence, it is
possible to derive other hyperbolic systems if wanted.


\section{New Models}
\label{sec:nhmsf}
In the last section, several conventional hyperbolic moment systems were
studied in the framework. As a powerful tool, the framework is not only
able to include existing models, but is also able to derive new
models. Based on the framework, we will derive some new hyperbolic
moment systems in this section.

\subsection{Regularization of Grad's ordered moment hierarchy}
For the conventional Boltzmann equation \eqref{eq:boltzmann} with a
Maxwellian as the weight function, there are two possible choices of the
subspace $\bbH_{sub, M}^{\weight^{[\bu,\theta]}}$, where
$\weight^{[\bu,\theta]}$ is the same as that in HME, and $M$ is a
positive integer. One choice is
\begin{equation}
    \bbH_{sub,M}^{\weight^{[\bu,\theta]}}=\rmspan\left\langle
    \weight^{[\bu,\theta]}\left\{\bxi^\alpha\right\}_{|\alpha|\leq M}\right\rangle
    =\rmspan\left\langle \{\mbH_{\alpha}\}_{|\alpha|\leq M} \right\rangle ,
\end{equation}
corresponding to $10$, $20$, $35$, $56$, $84$, $\dots$ moments or
moment systems $G10$, $G20$, $G35$, $G56$, $G84$, $\dots$ for
$D=3$. The moment systems in \cite{NRxx} and HME correspond to this
choice. This set of moments is sometimes called a \emph{full moment
  theory}, see e.g.  \cite{Torrilhon2014}, because it includes the
full set of moments up to order $M$.

The other choice is
\begin{equation}
    \begin{aligned}
    \bbH_{sub,M}^{\weight^{[\bu,\theta]}}&=\rmspan\left\langle
    \weight^{[\bu,\theta]}\{\bxi^\alpha\}_{|\alpha|\leq
    M-1}\bigcup
    \weight^{[\bu,\theta]}\{|\bxi|^2\bxi^{\alpha}\}_{|\alpha|=M-2}
    \right\rangle \\
    &=\rmspan\left\langle \left\{\mbH_{\alpha}\right\}_{|\alpha|\leq M-1}\bigcup
    \left\{ \sum_{d=1}^D\mbH_{\alpha+2e_d}\right\}_{|\alpha|=M-2}\right\rangle ,
    \end{aligned}
\end{equation}
corresponding to $5, 13, 26, 45, \dots$ moments or moment systems $G5,
G13, G26, G45, \dots$ for $D=3$. The $G13$ moment system in Section
\ref{sec:grad13} belongs to this class. The second set of moments can
be seen as a hierarchy of moment sets that is a kind of \emph{ordered
moment system}, because higher members of the hierarchy always include
fluxes of the lower members, see again \cite{Torrilhon2014} where this
notation is used first. Note that members of the ordered moment
hierarchy also have a rotationally invariant basis.

Full moment theories have been extensively studied and globally
hyperbolic versions for it have also been proposed. But for Grad's
ordered moment theories, there is only very few work, e.g.
\cite{Struchtrup}, and globally hyperbolic regularizations are only
proposed for G13. Here we give a concise derivation of the ordered
moment hierarchy and propose a globally hyperbolic version. Similar as
the definition of the regularized G13 moment system in Section
\ref{sec:grad13}, we only need to choose the projection. Hence, the
symbols $\weight^{[\bu,\theta]}$, $\mH$, $\bw$, $\bM_d$ and $\bD$ have
the same definitions as those in Section \ref{sec:hme}.

First, we define the moments
\[
    \Delta_{\alpha}=\frac{1}{2}\int_{\bbR^D}\frac{1}{\weight^{[\bu,\theta]}}
    f\sum_{d=1}^D\mbH_{\alpha+2e_d}\dd\bxi,\quad |\alpha|=M-2.
\]
Let $\{\mbH_{\alpha}\}_{\alpha\in\bbN^D}$ be the basis of
$\bbH^{\weight^{[\bu,\theta]}}$, and $\{\mbH_{\alpha}\}_{|\alpha|\leq
M-1}\bigcup \{ \sum_{d=1}^D\mbH_{\alpha+2e_d}\}_{|\alpha|=M-2}$
be the basis of $\bbH_{sub}^{\weight^{[\bu,\theta]}}$. Then
$\bPb=(p_{b,i,j})$ is, for $d=1,\cdots,D$,
\[
    p_{b,i,i}=1, \quad i=1,\dots,\mN( (M-1)e_D),\quad
    p_{b,\mN(\alpha+2e_1),\mN(\alpha+2e_d)}=1,
    \quad |\alpha|=M-2.
\]
Here $\mN( (M-1)e_D)$ is the cardinality of $\{\alpha\}_{|\alpha|\leq
M-1}$, and $\mN(\alpha+2e_1)$ is the consecutive number of
$\sum_{d=1}^D\mH_{\alpha+2e_d}$ in the basis of
$\bbH_{sub}^{\weight^{[\bu,\theta]}}$. The orthogonal projection is
used, thus $\bPp$ can be calculated based on \eqref{eq:bpp} as
\[
    \begin{aligned}
        &p_{p,i,i}=1, \quad i=1,\dots,\mN( (M-1)e_D),\\
        &p_{p,\mN(\alpha+2e_1),\mN(\alpha+2e_d)}=
        \frac{(\alpha+2e_d)!}{\sum_{d=1}^D(\alpha+2e_d)!},\quad
        |\alpha|=M-2, \text{ and } d=1,\dots,D,
    \end{aligned}
\]
where $\alpha!$ stands for $\prod_{d=1}^D\alpha_d!$.
Easy to check, we have $\bPp\bw=\bw_N$, where $N$ is the dimension of
$\bbH_{sub}^{\weight^{[\bu,\theta]}}$, and $\bw_N$ is
\[
    (\bw_N)_i = (\bw)_i, i=1,\dots,\mN( (M-1)e_D),\quad
    (\bw_N)_{\mN(\alpha+2e_1)} =
    \frac{\Delta_{\alpha}}{\sum_{d=1}^D(\alpha+2e_d)!}.
\]
Then
\[
    \bPp\bD\bPb^T\pd{\bw_{N}}{t}+\sum_{d=1}^3\bPp\bM_d\bD\bPb^T\pd{\bw_{N}}{x_d}
    =\bPp\bS,
\]
is Grad's ordered moment system of order $M$ and
\[
    \bPp\bD\bPb^T\pd{\bw_{N}}{t}+\sum_{d=1}^3\bPp\bM_d\bPb^T\bPp\bD\bPb^T\pd{\bw_{N}}{x_d}
    =\bPp\bS,
\]
is the regularized version thereof. Theorem \ref{thm:diagonalizable}
indicates that the moment system is globally hyperbolic.

As stated in Remark \ref{rm:dependence}, the matrixes $\bD$ and
$\bM_d$ and the vector $\bS$ in the upper equation is defined as
$\bD=\bD(\bPb\bPp\bw)$, $\bM=\bM(\bPb\bPp\bw)$,
$\bS=\bS(\bPb\bPp\bw)$, respectively.

Particularly, if $D=3$ and $M=2$, the moment system reduces to the
classical Euler equations, and if $D=3$ and $M=3$, the moment system
is that in Section \ref{sec:grad13}.

\subsection{Quadrature-based moment equations for multi-dimensional case}
\label{sec:QBMEMD}
QBME have been extended to the multi-dimensional case in
\cite{KoellermeierRGD2014}, based on the quadrature-based idea.
However, the tensor product approach for the quadrature points causes
that the resulting system in \cite{KoellermeierRGD2014} is not
rotationally invariant. Note that it is impossible to achieve
rotational invariance in that framework as there is no corresponding
rotational invariant Gaussian quadrature rule in multiple dimensions.

In this subsection, we extend QBME to the multi-dimensional case based
on the framework in Section \ref{sec:framework} to obtain a hierarchy
of globally hyperbolic and rotationally invariant moment systems.

For the $D$-dimensional Boltzmann equation, the kinetic equation is
\[
    \begin{aligned}
        &\bdeta_1=(u_1,\dots,u_D,\theta),\quad
        \bv(\bxi)=\frac{\bxi-\bu}{\sqrt{\theta}},\quad
        p_d(\bv)=u_d+\sqrt{\theta}v_d,\quad\\
        &\mL\left( \pd{}{s};f, \bdeta_1, \bv(\bxi) \right)=\pd{f}{s}-\sum_{k=1}^D\pd{f}{v_k}
        \left( \frac{1}{\sqrt{\theta}}\pd{u_k}{s}
        +\frac{1}{2\theta}v_k\pd{\theta}{s}
        \right),\; s=t, x_d,
    \end{aligned}
\]
where $f=f(t,\bx,\bv)$. The weight function and the orthogonal
weighted polynomials are defined by
\[
    \weight(\bv)=\frac{1}{\sqrt{2\pi}^D}\exp\left( -\frac{|\bv|^2}{2}
    \right),\quad
    \mh_{\alpha}(\bv)=(-1)^{|\alpha|}\od{^\alpha\weight}{\bv^\alpha},
    \quad \alpha\in\bbN^D,
\]
and satisfy the following properties:
\begin{itemize}
    \item Differential relation:
        $\od{\mh_\alpha(\bv)}{v_d}=-\mh_{\alpha+e_d}(\bv)$,
        $d=1,\dots,D$,
    \item Recurrence relation:
        $\mh_{\alpha+e_d}(\bv)=v_d\mh_{\alpha}-\alpha_d\mh_{\alpha-e_d}(\bv)$,
        $d=1,\dots,D$.
\end{itemize}
Similar as in Section \ref{sec:QBME}, we define $\bD_{v,d}$ such that
$\od{\boldsymbol{\mh}}{v_d}=-\bD_{v,d}^T\boldsymbol{\mh}$ and
$\bM_{v,d}$ such that
$v_d\boldsymbol{\mh}=\bM_{v,d}^T\boldsymbol{\mh}$, $d=1,\dots,D$,
where $\boldsymbol{\mh}=(\mh_{\alpha})$ is a vector of elements sorted
by ascending order of $\alpha$. We set $\bdeta=(u_1,\dots,u_D,\theta)$
and some calculations yield the constraints
\[
    f_{e_d}=0, d=1,\dots,D,\quad
    \sum_{d=1}^D{f_{2e_d}}=0.
\]
Hence, we use $u_i$ to replace $f_{e_i}$ and $\theta/2$ to replace
$f_{2e_1}$ in $\bdf$, and name the resulting vector $\bw$, where
$\bdf=(f_{\alpha})$ is a vector of elements sorted by ascending order
of $\alpha$. We choose a positive integer $M\geq3$, and the subspace is
then defined as
$\bbH^{\weight}_{sub}=\rmspan\left\langle\left\{\mh_{\alpha}
\right\}_{|\alpha|\leq M}\right\rangle$. Note that this yields a
rotationally invariant basis, in contrast to the approach of the
existing multi-dimensional QBME method.

The projection operator is chosen as the orthogonal projection, i.e.
$\bPb=\bPp=\boldsymbol{\mathrm{T}}$.

For the time and space derivative, we have, for $s=t,x_d$,
$d=1,\dots,D$,
\[
    \begin{aligned}
        \mL\left( \pd{}{s}; f,\bdeta_1,\bv\right)
        &= \left\langle \boldsymbol{\mh}, \pd{\bdf}{s}\right\rangle_{\infty}
        -\sum_{k=1}^D\left\langle \od{\boldsymbol{\mh}}{v_k}, \bdf \left(
        \frac{1}{\sqrt{\theta}}\pd{u_k}{s} +\frac{1}{2\theta}v_k\pd{\theta}{s}
        \right)\right\rangle_{\infty} \\
        &= \left\langle \boldsymbol{\mh}, \pd{\bdf}{s}\right\rangle_{\infty}
        +\sum_{k=1}^D\left\langle \boldsymbol{\mh}, \bD_{v,k}\bdf
        \frac{1}{\sqrt{\theta}}\pd{u_k}{s} \right\rangle_{\infty}
        +\sum_{k=1}^D\left\langle \boldsymbol{\mh},\bM_{v,k}\bD_{v,k}\bdf \frac{1}{2\theta}\pd{\theta}{s}
        \right\rangle_{\infty}.
    \end{aligned}
\]
Similar as for the 1D case, the internal projection strategy $PS_1$ is
\[
    \begin{aligned}
    \mL^{PS_1}\left( \pd{}{s}; f,\bdeta_1,\bv\right)
        = \left\langle \boldsymbol{\mh}, \bPb^T\bPp\pd{\bdf}{s}\right\rangle_{\infty}
        &+\sum_{k=1}^D\left\langle \boldsymbol{\mh}, \bD_{v,k}\bPb^T\bPp\bdf
        \frac{1}{\sqrt{\theta}}\pd{u_k}{s} \right\rangle_{\infty} \\
        &+\sum_{k=1}^D\left\langle \boldsymbol{\mh},
        \bM_{v,k}\bPb^T\bPp\bD_{v,k}\bPb^T\bPp\bdf
        \frac{1}{2\theta}\pd{\theta}{s}
        \right\rangle_{\infty}.
    \end{aligned}
\]
Collecting all the coefficients of $\pd{\bw}{s}$, we obtain
$\bPp\bD_{PS_1}\bPb^T=(d_{ps,i,j})_{N\times N}$, $N=\mN(Me_D)$,
satisfying
\[
    \begin{aligned}
        &d_{ps,\mN(\alpha),\mN(\alpha)}=1, |\alpha|\neq 1 \text{ and }
        \alpha\neq 2e_1,\quad
        d_{ps, \mN(\alpha),
        \mN(e_k)}=\frac{f_{\alpha-e_k}}{\sqrt{\theta}}, |\alpha|\leq
        M, k=1,\dots,D,\\
        &d_{ps,\mN(\alpha),\mN(2e_1)}=\frac{1}{2\theta}\sum_{k=1}^D(f_{\alpha-2e_k}+(\alpha_k+1)f_{\alpha}),
        |\alpha|\leq M-1,\\
        &d_{ps,\mN(\alpha),\mN(2e_1)}=\frac{1}{2\theta}\sum_{k=1}^Df_{\alpha-2e_k},
        |\alpha|=M,\quad
        d_{ps,\mN(2e_1),\mN(2e_k)}=-1, k=1,\dots,D,
    \end{aligned}
\]
where all entries not defined above are zero. Based on the analysis in
Section \ref{sec:hme} and Section \ref{sec:QBME}, it is easy to verify that
$\bPp\bD_{PS_1}\bPb^T$ is invertible.

Since $p_d(\bv)=u_d+\sqrt{\theta}v_d$, the internal projection strategy $PS_2$
vanishes and $\bM_d=u_d\identity +\sqrt{\theta}\bM_{v,d}$.
Since the projection $\mP$ is an orthogonal projection and $p_d(\bv)$ is
a linear polynomial, Theorem \ref{thm:diagonalizable} indicates the
resulting system
\[
    \bPp\bD_{PS_1}\bPp^T\pd{\bw}{t}+\sum_{d=1}^D\bPp\bM_d\bPb^T\bPp\bD_{PS_1}\bPb^T\bPp\pd{\bw}{x_d}=\bPp\bS
\]
is globally hyperbolic and a rotationally invariant, multi-dimensional extension of QBME.

We emphasize that this extension is only possible with the help of the
operator projection approach. In multiple dimensions, there is no
Gaussian quadrature rule that could result in a rotationally
invariant moment system. However, the use of the projection operator
$\bPb=\bPp=\boldsymbol{\mathrm{T}}$ mimics the effect of a Gaussian
quadrature rule, as it essentially cuts off the highest order term
during every different step of the derivation. We can therefore say that
the derivation of the new system follows the quadrature-based
technique but uses an operator projection to achieve rotational
invariance. 

\section{Conclusion} \label{sec:conclusion}
For first-order convection equations, hyperbolicity is necessary
for the existence of a solution. Historically, the lack of global
hyperbolicity has been a critical defect of \Grad moment method, and
largely limited the development of moment methods. In this paper, we
investigate \Grad moment system and its globally hyperbolic
regularized version for the 1D Boltzmann equation, then point out that the most
essential point of the regularization is to treat the time derivative
and the space derivative in the same manner.

Based on this observation, a
general framework for the construction of hyperbolic moment systems from
kinetic equations using the operator projection method is proposed.
This framework is so concise and clear that it can be treated as an
algorithm, and once the four inputs, i.e. the kinetic
equation, the weight function, the projection operator and the
internal projection strategy, are given, the moment system can be derived with
some routine calculations. Among the four inputs, the weight function is
the most essential one, because it determines the approximation space. The
projection operator determines the type of the moment system.
In this framework, it is possible to contain some information of the
problems to be solved in the moment system by the choice of an appropriate weight
function, and it is also possible to derive moment systems without the
projection first and then to perform the projection at last, which helps to
understand the difference of moment systems with the same weight
function (such as G20 and G13) or even the same basis (such as 1D HME and QBME).

Different existing hyperbolic models, such as hyperbolic regularizations of
\Grad moment method for 1D (Section  \ref{sec:hme1d}) and $n$D (Section
\ref{sec:hme}), anisotropic hyperbolic moment equations (Section
\ref{sec:ahme}), the hyperbolic version of the G13 moment system
(Section \ref{sec:grad13}), Levermore's maximum entropy principle
(Section \ref{sec:Levermore}) and quadrature-based moment equations
(QBME) (Section \ref{sec:QBME}), are included in the framework.
Actually, some other models, such as the $P_N$ and $M_N$ model in
radiative transfer are also included in this framework. Furthermore,
based on the framework, we propose a hyperbolic regularization of the
ordered moment hierarchy (such as 13, 26, 45 moment systems), and
extend QBME to the multi-dimensional case with the resulting moment
system being rotational invariant.

The aforementioned examples and applications thus show the benefit of
the new operator projection approach and open many new possibilities
for research on moment methods.

\section*{Acknowledgements}

\bibliographystyle{plain}
\bibliography{article}
\end{document}